\documentclass[a4paper,12pt]{article}

\usepackage{cite}
\usepackage{graphicx}
\usepackage{latexsym}
\usepackage{multirow}
\usepackage{algorithm,algorithmic}
\usepackage{amssymb} 
\usepackage{multirow,eepic}
\usepackage{wrapfig}
\usepackage{lipsum}

\makeatletter
\def\senbun#1(#2)#3({\@senbun(#2)(}
\def\@senbun(#1,#2)(#3,#4){%
   \@tempdima#1\p@ \advance\@tempdima#3\p@
   \divide\@tempdima\tw@
   \@tempdimb#2\p@ \advance\@tempdimb#4\p@
   \divide\@tempdimb\tw@
   \edef\@senbun@temp{\noexpand\qbezier(#1,#2)%
      (\strip@pt\@tempdima,\strip@pt\@tempdimb)(#3,#4)}%
   \@senbun@temp}
\makeatother

\newcommand{\N}{{\rm I\kern-.22em N}} 
\newcommand{\Z}{{\sf Z\kern-.42em Z}} 
\newcommand{\R}{{\rm I\kern-.22em R}} 
\newcommand{\BbbC}{{\rm\kern.22em\rule[.1ex]{.06em}{1.4ex}\kern-.28em C}} 
\newcommand{\BbbQ}{{\rm\kern.22em\rule[.1ex]{.06em}{1.4ex}\kern-.28em Q}}

\newcounter{Codeline}
\newcommand{\Newcodeline}{\setcounter{Codeline}{1}}
\newcommand{\Cl}{{\theCodeline}: \addtocounter{Codeline}{1}}
\newcommand{\crm}{\\}

\newcommand{\overarc}[1]{\stackrel{\rotatebox{90}{\big)}}{#1}}

\usepackage{amsthm}
\theoremstyle{definition}
\newtheorem{theorem}{Theorem}
\newtheorem*{theorem*}{Theorem}
\newtheorem{lemma}{Lemma}
\newtheorem*{lemma*}{Lemma}
\newtheorem{corollary}{Corollary}
\newtheorem*{corollary*}{Corollary}

\newtheorem*{definition*}{Definition}

\begin{document}

\title{Gathering Problems for Autonomous Mobile Robots with Lights
}




\author{Satoshi Terai\footnote{Graduate School of Science and Engineering, Hosei University, Tokyo, 184-8485, Japan, satoshiterai7648@gmail.com}, 
Koichi Wada\footnote{Faculty of Science and Engineering, Hosei University, Tokyo, 184-8485, Japan, wada@hosei.ac.jp},
Yoshiaki Katayama\footnote{Graduate School of Engineering, Nagoya Institute of Technology,Nagoya, 466-8555, Japan, katayama@nitech.ac.jp}}



\maketitle

\begin{abstract}
We study the {\em Gathering} problem for $n$ autonomous mobile robots in semi-synchronous settings with persistent memory called {\em  light}.
It is well known that Gathering is impossible in a basic model when robots have no lights,
if the system is semi-synchronous or even centralized (only one robot is active in each time) \cite{SY,DGCMR}.  
On the other hand,  
Rendezvous (Gathering when $n=2$) is possible if robots have lights of various types with a constant number of colors \cite{FSVY,V}.
If robots can observe not only their own lights but also other robots' lights, their lights are called {\em full-light}.
If robots can only observe the state of other robots' lights, their lights are called {\em external-light}.
If robots can only observe their own lights, their lights are called {\em internal-light}.

In this paper, we extend the model of robots with lights so that Gathering algorithms can be discussed properly.
Then we show Gathering algorithms with three types of lights in the semi-synchronous settings
and reveal relationship between the power of lights and other additional assumptions.
The most algorithms shown here are optimal in the number of colors they use.

\end{abstract}

\section{Introduction}
\noindent{\bf Background and Motivation}\hspace*{1em} 
The computational issues of autonomous mobile robots have been research object in distributed computing fields.
In particular, a large amount of work has been dedicated to the research of theoretical models of autonomous mobile robots
\cite{AP,BDT,CFPS,DKLMPW,IBTW,KLOT,SDY,SY}.
In the basic common setting, a robot is modeled as a point in a two dimensional plane and its capability is quite weak.
We usually assume that robots are {\em oblivious} (no memory to record past history), {\em anonymous}  and {\em uniform} (robots have no IDs and run identical algorithms)\cite{FPS}.
Robots operate in Look-Compute-Move (LCM) cycles in the model. In the Look operation, robots obtain a snapshot of the environment and
they execute the same algorithm with the snapshot as an input in Compute operation, and move towards the computed destination in Move operation.
Repeating these cycles, all robots perform a given task.
It is difficult for these too weak robot systems to accomplish the task to be completed. 
Revealing the weakest capability of robots to attain a given task is one of the most interesting challenges 
in the theoretical research of autonomous mobile robots.

In this paper, we also explore such weakest capabilities. In particular,
we reveal the weakest additional assumptions for the task which cannot be solved in the basic common models. 
The problem considered in this paper is {\em Gathering}, which is one of the most fundamental tasks of autonomous mobile robots. Gathering is the process of $n$ mobile robots, initially located on arbitrary positions, meeting within finite time at a location, not known a priori. When there are two robots in this setting, this task is called {\em Rendezvous}. 
Since Gathering and Rendezvous are simple but essential problems, they have intensively studied  
and a number of possibility and/or impossibility results have been shown under the different assumptions\cite{AP,BDT,CFPS,DGCMR,DKLMPW,DP,FPSW,IKIW,ISKIDWY,KLOT,LMA,P,SDY}.
The solvability of Gathering and Rendezvous depends on the activation schedule and the synchronization level. 
Usually three basic types of schedulers are identified, the fully synchronous (FSYNC), the semi-synchronous (SSYNC) and the asynchronous (ASYNC).
In the FSYNC model, there is a common round and in each round all robots are activated simultaneously 
and Compute and  Move are done instantaneously. 
The SSYNC is the same as FSYNC except that in each round only a subset of robots are activated. 
In the ASYNC scheduler, there is no restriction about notion of time, 
Compute and Move in each cycle can take an unpredictable amount of time, 
and the time interval between successive activations is also unpredictable (but these times must be finite).
Gathering and Rendezvous are trivially solvable in FSYNC and the basic model.
However, these problems can not be solved in SSYNC without any additional assumptions \cite{FPS}. 
In particular, Gathering can not be solvable even in a restricted subclass of SSYNC scheduler, 
where exactly one robot is activated in each round (called CENT) \cite{DGCMR}.
If all robots are initially located on different positions (called distinct Gathering), 
this Gathering can not be solved even in the ROUND-ROBIN scheduler, 
in which exactly one robot is activated in each round and always in the same order \cite{DGCMR}.

In \cite{DFPSY}, persistent memory called {\em light} has been introduced to reveal relationship between the solvability of Gathering and the synchrony of schedulers and they show asynchronous robots with lights equipped with a constant number of colors, are strictly more powerful than semi-synchronous robots without lights.
Robots with lights have been also introduced 
in order to solve Rendezvous without any other additional assumptions. \cite{FSVY,DFPSY,V}.
Table~\ref{tab:Table-Rendezvous} shows results to solve Rendezvous by robots with lights in each scheduler, 
where the circle ($\bigcirc$) and the cross ($\times$) mean Rendezvous is solvable and unsolvable, respectively, the hyphen ($-$) indicates that this part has been solved under weaker conditions or unsolved under the stronger ones, and the question mark ($?$) means that this part has not been solved. 
In the table,
{\em full-light} means that robots can see not only lights of other robots but also their own light, 
and {\em external-light} and {\em internal-light}\footnote{In \cite{FSVY}, external-light and internal-light are called FCOMM and FSTATE, respectively.}  mean
that they can see only lights of other robots and only own light, respectively. 
Full-light is not a weaker assumption than external-light and internal-light,
and internal-light seems to be weaker than external-light. 
Although this relationship is not proved,  the results indicate the relationship.
For example, Rendezvous is solved by robots with $3$ colors of external-light
in SSYNC and non-rigid. On the other hand, it is solved by robots with $6$ colors of internal-light in SSYNC and rigid, 
where robots can reach the computed destination in {\em rigid} and 
they may be stopped before reaching it in {\em non-rigid}.
The number of colors used in the external-light  algorithm is less than that in the internal-light one, and
the former uses non-rigid assumption but
the latter uses the stronger one, rigidness. When robots know the value $\delta$ of a minimum distance movement  in non-rigid (denoted by $+\delta$ in the table),
the number of colors can be reduced into only $3$ in SSYNC and internal-light. Thus, the assumption robots know $\delta$ in non-rigid
seems to be stronger than that in rigid.
The power of lights to solve other problems are discussed in \cite{LFCPSV,DFN}.

{\small
\begin{table}[h]
\centering
\caption{Rendezvous algorithms by robots with lights.}
\label{tab:Table-Rendezvous}
\begin{tabular}{|c|c|c|c|c|c|}
\hline
scheduler      & movement  & full-light & external-light & internal-light & no-light \\ \hline\hline
FSYNC & $-$ & $-$ & $-$ & $-$ & $\bigcirc$\cite{FPS}  \\ \hline
\multirow{2}{*}{SSYNC} & rigid     & $-$ & ?        & 6\cite{FSVY}        & \multirow{2}{*}{$\times$\cite{FPS}} \\ \cline{2-5} 
               & non-rigid & $-$    & 3\cite{FSVY}          & 3(+$\delta$)\cite{FSVY}   &     \\ \hline
\multirow{2}{*}{ASYNC}          & rigid     & $-$    & 12\cite{FSVY}         & ?        & \multirow{2}{*}{$-$}  \\ \cline{2-5} 
               & non-rigid & 2\cite{HDT}    & 3(+$\delta$)\cite{FSVY}   & ?      &     \\ \hline
\end{tabular}
\end{table}
}

\noindent{\bf Our Contribution}\hspace*{1em}
In this paper, we discuss Gathering algorithms by robots with lights and reveal relationship about the solvability  between kinds of lights and other additional assumptions.
First, we extend the model of robots with lights in which Gathering algorithms can be discussed properly.
Unlike Rendezvous, the multiplicity detection affects the solvability of Gathering and in fact, 
Gathering can be solvable by robots without lights if robots detect the exact number of robots 
on the same points (strong multiplicity detection) \cite{FPS}. 
When several robots with lights occupy a same location, if we assume that robots can  
recognize all colors of the robots on the location, the strong multiplicity detection is possible and robots without lights 
can solve the Gathering problem. Thus, we define the model 
so that robots can not detect the exact number of robots on the same location
provided that robots use only one color.  
We consider  two types of models according to views robots observe.
One is called {\em set-view},
where robots can recognize sets of colors robots have on the same locations\footnote{Recognition of multi-set-view enables us the strong multiplicity detection.}. 
Another is called {\em arbitrary-view}, where robots can recognize arbitrary color robots have on the same locations. 
We usually use the set view assumption in this paper. 
We also consider the case that several robots occupy a same location and how these robots observe one another. 
We call {\em local-aware}
if any robot can recognize other robots located on the same position 
and {\em local-unaware}, otherwise. If we assume the local-awareness, the view of robots obeys the model (set-view or arbitrary-view) we assume. We usually assume the local-unawareness 
but we show the local-awareness reduces the number of colors for some Gathering algorithm.  

In the extended model, we show the following six Gathering algorithms with lights (Table~\ref{tab:Table-Gathering}). 
\begin{enumerate}
\item[(a)] {\bf Algorithm~1}:(2, full, SSYNC, non-rigid),
\item[(b)] {\bf Algorithm~2}:(3, external, SSYNC, rigid),
\item[(c)] {\bf Algorithm~3}:(2, external, SSYNC, rigid; local-awareness),
\item[(d)] {\bf Algorithm~4}+{\bf 5}+{\bf 6}:(2, internal, SSYNC, non-rigid; D-distant),
\item[(e)] {\bf Algorithm~7}:(2, external, CENT, non-rigid), and
\item[(f)] {\bf Algorithm~8}:(2, internal, ROUND-ROBIN, rigid),
\end{enumerate}
where Algorithm i:($\alpha, \beta, \gamma, \delta ; \eta$) means Algorithm i uses $\beta$-light with $\alpha$ colors  and works in $\gamma$ scheduler and movement restriction $\delta$ with additional assumption $\eta$ (if any).

Table~\ref{tab:Table-Gathering} also shows the unsolvability of Gathering and 
every algorithm using two colors is optimal with respect to the number of colors, since Gathering is not solvable
without lights in every case.
We show four algorithms in SSYNC.
In the full-light model, we construct an algorithm with the weakest assumptions, non-rigid and the least number of colors ($2$ colors).
In the external-light model,  we construct  an algorithm with $3$ colors of lights and rigid assumption and we reduce the number of colors used in the algorithm
into $2$ colors if the local-awareness is assumed. 
In the internal-light model, 
with a little bit more assumption we construct an algorithm with lights of $2$ colors if robots know the minimum distance $\delta$ of movement in non-rigid and if initial configurations of robots satisfy some condition called $2\delta$-distant, where arbitrary two robots not occupying the same locations
are separated by at least $2\delta$.
We reveal some relationship between view of lights (full, external and internal) and the other assumptions (movement, local-awareness) in our algorithms. 

We also construct Gathering algorithms in CENT and ROUND-ROBIN schedulers.
Since Gathering  and distinct Gathering can not be solved without lights in CENT and ROUND-ROBIN, respectively \cite{DGCMR},
our algorithms use  $2$ colors but fairly weak assumptions. We give  algorithms in external-light of $2$ colors, non-rigid and CENT,
and in internal-light of $2$ colors, rigid and ROUND-ROBIN, respectively. 
 
{\small
\begin{table}[h]
\centering
\caption{Our Gathering algorithms by robots with lights.}
\label{tab:Table-Gathering}
{\footnotesize
\begin{tabular}{|c|c|c|c|c|c|}
\hline
scheduler      & movement  & full-light & external-light & internal-light & no-light \\ \hline\hline
FSYNC & $-$ & $-$ & $-$ & $-$ & $\bigcirc$  \\ \hline
ROUND-ROBIN & rigid & $-$ &  $-$ & 2(f) & $\times^{*}$ \cite{DGCMR} \\ \hline
CENT & non-rigid & $-$ & 2(e) & ? & $\times$\cite{DGCMR} \\ \hline
\multirow{2}{*}{SSYNC} & rigid & $-$ & 3(b),2$^{**}$(c) & ? & \multirow{2}{*}{$-$} \\ \cline{2-5} 
               & non-rigid & 2(a) & ? & 2(+$\delta$)$^{***}$(d) &  \\ \hline
\end{tabular}

\smallskip

\noindent
$^*$Distinct Gathering, 
$^{**}$local-awareness, 
$^{***}2\delta$-distant
}
\end{table}
}

\medskip

The remainder of the paper is organized as follows. In Section
~\ref{sec:model}, we define a robot model with lights, gathering problems, and terminologies. 
Section~\ref{sec:impo} shows 
the previous results for Gathering and Rendezvous problems,
and Section~\ref{sec:GatheringAlgorithms} shows Gathering algorithms of robots with lights in SSYNC 
on several situations for which Gathering algorithms do not exist by using lights. Section~\ref{sec:CSET} shows Gathering algorithms of robots with lights in CENT.  Section~\ref{sec:conclusion}
concludes the paper.

\section{Model and Preliminaries}\label{sec:model}

We consider a set of anonymous mobile robots $\mathcal{R} = \{ r_1, \ldots, r_n \}$ located in $\R^2$.
Each robot $r_i$ has a persistent state $\ell_i$  called {\em  light} which may be taken from a finite set of colors $L$. 

We denote by $\ell_i(t)$ the color that the light of robot $r_i$ has at time $t$ and $p_i(t) \in \R^2$ the position occupied by  $r_i$ at time $t$ represented in some global coordinate system. A {\em configuration} $\mathcal{C}(t)$ at time $t$ is a multi-set of $n$ pairs $(\ell_i(t),p_i(t))$, each defining the color of light and the position of robot $r_i$ at time $t$. 
When no confusion arises, $\mathcal{C}(t)$ is simply denoted by $\mathcal{C}$.

For a subset $S$ of $L \times \R^2$, $\mathcal{L}(S)$ and $\mathcal{P}(S)$ are denoted as projections to $L$ and $\R^2$ from $S$, respectively.

Each robot $r_i$ has its own coordinate system where $r_i$ is located at its origin at any time. These coordinate systems do not necessarily agree with those of other robots. It means that there is no common knowledge of unit of distance, directions of its coordinates,  or clockwise orientation ({\em chirality}
\footnote{In some cases, we may assume chirality.}).

At any point of time, a robot can be active or inactive. When a robot $r_i$ is activated, it executes 
$\mathit{Look}$, $\mathit{Compute}$, and $\mathit{Move}$ cycles:
\begin{itemize}
\item {\bf Look:} The robot $r_i$ activates its sensors to obtain a snapshot which consists of a pair of light and position for every robot with respect to the coordinate system of $r_i$. 
Let $\mathcal{SS}_i(t)$ denote the snapshot of $r_i$ at time $t$. 
We assume robots can observe all other robots(unlimited visibility). Note that $\mathcal{SS}_i(t)$ represents a sub-multi-set of $\mathcal{C}(t)$ according to imposed assumptions by the local coordinate system of $r_i$, where $r_i$ is at the origin.
\item {\bf Compute:} The robot $r_i$ executes its algorithm using the snapshot and the color of its own light (if allowed by the model) and returns a destination point $des_i$ expressed in its coordinate system and a light $\ell_i \in L$. The robot $r_i$ sets its own light $\ell_i$ to the color.
\item {\bf Move:} The robot $r_i$ moves to the computed destination $des_i$.
A robot $r$ is said to {\em collide} with robot $s$ at time $t$ if $p(r,t)=p(s,t)$ and at time $t$ $r$ is performing $\mathit{Move}$. The collision is {\em accidental} if $r$'s destination is not $p(r,t)$.
Since robots are seen as points, we assume that accidental collisions are immaterial. A moving robot,
upon causing an accidental collision, proceeds in its movement without changes, in a ``hit-and-run'' fashion \cite{FPS}.
If the robot may be stopped by an adversary before reaching the computed destination,
the movement is said to be {\em non-rigid}. Otherwise,
it is said to be {\em rigid}. 
If stopped before reaching its destination, we assume that a robot moves at least a minimum distance $\delta >0$. Note that without this assumption an adversary could make it impossible for any robot to ever reach its destination. If the distance to the destination is at most $\delta$, the robot can reach it. If the movement is non-rigid and robots know the value of $\delta$, it is called {\em non-rigid($+\delta$)}. 
\end{itemize}  

In ${\mathit Look}$ operation, a snapshot $\mathcal{SS}_i$ of $r_i$ should contain positions of all robots including $r_i$. However, if other robots are located on $p_i$
and $r_i$ can recognize the other robots, robots have somehow multiplicity detection at this point. Thus,
we need some treatment of observation of other robots located on $p_i$ for robot $r_i$.
If any robot $r_i$ can observe other robots located on $p_i$, it is said to be {\em local-aware}.
Otherwise, it is said to be {\em local-unaware}. Note that if we assume the local-awareness, $r_i$ recognizes whether other robots occupy the location $p_i$ or not. In the following, we usually use the local-unaware assumption but we will show that 
we can reduce the number of colors
for some Gathering algorithm if we assume the local-awareness.

A scheduler decides which subset of robots is activated for every configuration. 
The schedulers we consider are asynchronous and semi-synchronous and it is assumed that schedulers are {\em fair}, each robot is activated infinitely often.
\begin{itemize}
\item {\bf ASYNC:} 
The asynchronous scheduler (ASYNC), activates the robots independently, and the duration of each $\mathit{Compute}$, $\mathit{Move}$ and between successive activities is finite and unpredictable. As a result, robots can be seen while moving and the snapshot obtained with the Look operation and its actual configuration when performing the $\mathit{Compute}$ operation are not the same and so its computation may be done with the old configuration.
\item {\bf SSYNC:} 
The  semi-synchronous scheduler (SSYNC) activates a subset of all robots synchronously  and their $\mathit{Look}$-$\mathit{Compute}$-$\mathit{Move}$ cycles are performed at the same time. We can assume that activated robots at the same time obtain the same snapshot and their $\mathit{Compute}$ and $\mathit{Move}$ are executed instantaneously.
In SSYNC, we can assume that each activation defines discrete time called {\em round} and $\mathit{Look}$-$\mathit{Compute}$-$\mathit{Move}$ is performed instantaneously in one round. In the following, since we consider SSYNC and its subsets, we use round and time interchangeably.
\end{itemize}

As a special case of SSYNC, if all robots are activated in each round,
the scheduler is called fully-synchronous (FSYNC).
We also consider the following subsets of SSYNC.
\begin{itemize} 
\item {\bf CENT:} 
In CENT scheduler exactly one robot is activated in each round.
\item {\bf $k$-BOUNDED:} 
In $k$-BOUNDED scheduler there exists some bound $k$ such that between any two consecutive activation of any robot, no other robot is activated more than $k$ times.
\item {\bf ROUND-ROBIN:} 
ROUND-ROBIN scheduler is $1$-BOUNDED and CENT. It means that robots are activated one in each round  and always in the same sequence. 
\end{itemize}

Let $\mathcal{C}(t)$ be a configuration at round $t$.
When $\mathcal{C}(t)$ reaches $\mathcal{C}(t+1)$ by executing the cycle at $t$, it is denoted as $\mathcal{C}(t) \rightarrow \mathcal{C}(t+1)$.
The reflective and transitive closure is denoted as $\rightarrow^*$.

Snapshots may be different by using assumptions even if these configurations are the same, and they depend on multiplicity detection or how robots can see lights of other robots when robots equip with lights.
The robots are said to be capable of {\em multiplicity detection} whether they can distinguish if a point is occupied  by at least two robots. The multiplicity detection is {\em strong} if the robots can detect the exact number of robots on the same points.

In our settings, robots have  persistent lights and can change its color after $\mathit{Compute}$ operation. We consider the following three robot models according to the visibility of lights.
\begin{itemize}
\item {\em full-light}, the robot can recognize not only colors of lights of other robots but also its own color of light.
\item {\em external-light}, the robot can recognize only colors of lights of other robots but cannot see its own color of light. Note robot can change its own color.
\item {\em internal-light}, the robot can recognize only its own color of light but cannot see colors of lights of other robots.
\end{itemize}

When a robot performs $\mathit{Look}$ operation in the internal-light model,  its snapshot is the same as that in the case of robots without lights. On the other hand, in the full-light or external-light model,  it obtains a snapshot with locations of other robots and their colors. In this case we consider several types of snapshots according to view robots observe.    

Given a snapshot $\mathcal{SS}_i$ of a robot $r_i$  and a point $p_j(j \neq i)$ included in $\mathcal{P}(\mathcal{SS}_i)$,
a view  $V_i[p_j]$ of $p_j$ in $\mathcal{SS}_i$ is a subset of $AL_i[p_j]=\{ \ell | (\ell, p_j) \in  \mathcal{SS}_i, r_j \neq r_i \}$, where  $AL_i[p_j]$ is a multi-set  of colors of other robots that $r_i$ can see at point $p_j$
\footnote{Note that $\ell_i$ is not in $AL_i[p_j]$, and
$AL_i[p_j]$ is not defined for $p_j$ such that $r_j$ is located at $p_i$ under the local-unawareness.}.
For any robot $r_i$ and any point $p$ of snapshot of  $r_i$, if $V_i[p]=AL_i[p]$, view of robots is called {\em multi-set- view}. 
If $V_i[p]$ regards $AL_i[p]$ as just a set, it is called {\em set-view}. 
If $V_i[p]$ is a set of any single element taken from  $AL_i[p]$, it is called {\em arbitrary-view}.
Let $V_i$ denote $\cup_{(\ell,p) \in \mathcal{SS}_i}V_i[p]$.

Multi-set-view is a strong assumption, because robots without lights (with one color) can have strong multiplicity detection if multi-set-view is assumed. In fact we can solve a Gathering problem by using robots without lights and multi-set-view\cite{FPS}. On the other hand, set-view and arbitrary-view do not imply  multiplicity detection.
In the following we assume set-view or arbitrary-view.

An $n$-{\em Gathering} problem is defined as follows: given $n (\geq 2)$ robots initially placed at arbitrary positions  in $\R^2$,
they congregate in finite time at a single location which is not predefined.
In the following, the case $2$-Gathering problem is called {\em Rendezvous} and the $n$-Gathering  problem ($n \geq 3)$ is simply called Gathering.
Gathering is said to be {\em distinct} if all robots are initially placed in different positions.  
An algorithm solving Gathering is said to be {\em self-stabilizing} if robots are initially set their lights to arbitrary colors and they start their execution from $\mathit{Look}$ operation. 

Let $\mathcal{S}$ be a configuration or a snapshot.
Given $\mathcal{S}$, 
let $\mathcal{H}(\mathcal{S})$ be the convex hull defined by $\mathcal{S}$, let ${\partial}\mathcal{H}(\mathcal{S})$ denote the set of robots on the border of $\mathcal{H}(\mathcal{S})$, and let $\mathcal{I}(\mathcal{S})$ the set of robots that are interior of $\mathcal{H}(\mathcal{S})$. 

Given two points $p,q \in \R^2$, we indicate the line segment by $\overline{pq}$ and its length by $|\overline{pq}|$. 
Given a snapshot $\mathcal{SS}$,  $\mathit{SEC}(\mathcal{SS})$ denotes the smallest enclosing circle containing  $\mathcal{P}(\mathcal{SS})$, and the length of its diameter and its center are denoted by $\mathit{Diam}(\mathcal{SS})$ and   $\mathit{CTR}(\mathcal{SS})$, respectively.
A longest distance segment (\textit{LDS}, for short)  in $\mathcal{SS}$ is a line segment  $\overline{pq}$ such that $p,q \in  \mathcal{P} (\mathcal{SS})$ and  $|\overline{pq}|=max_{x,y \in \mathcal{P} (\mathcal{SS})}|\overline{xy}|$ and a set of
the longest distance segments in $\mathcal{SS}$ is denoted by $\mathit{LDS}(\mathcal{SS})$.

We use the following geometric properties about  
the smallest enclosing circles and the longest distance segments to obtain a unique longest distance segments.

Let $\mathcal{S}$ be a configuration or a snapshot.

\begin{lemma}\label{lemma:LDSofCH1}
For any  $\overline{pq} \in \mathit{LDS}(\mathcal{S})$,
$|\overline{pq}|  \leq \mathit{Diam}(\mathcal{S})$.
\end{lemma}
\begin{proof}
The two points $p$ and $q$ exist in $\mathit{SEC}(\mathcal{S})$.
\end{proof}

\begin{lemma}\label{lemma:LDSofCH2}
If $\overline{pq} \in \mathit{LDS}(\mathcal{S})$, then $p,q \in \mathcal{H}(S)$.
\end{lemma}
\begin{proof}
Suppose that $p$ or $q$ is not a vertex of  $\mathcal{H}(S)$. Let $\overline{p_1p_2}$ and $\overline{q_1q_2}$ be
the nearest edges of $\mathcal{H}(S)$ from $p$ and $q$, respectively. 
Note that $p$ and/or $q$ can be on the edges not vertices.
Comparing $\overline{pq}$ and the convex quadrilateral (or the triangle) composed of 
$\overline{p_1p_2}$ and $\overline{q_1q_2}$, the longer diagonal (or the longest edge) is longer than $\overline{p_1p_2}$.
It contradicts that $\overline{pq}$ is the longest distance segment.
\end{proof}

\begin{lemma}\label{lemma:LDSofCH3}
If for any $\overline{pq} \in \mathit{LDS}(\mathcal{S})$, both  $p$ and $q$ are located on $\mathit{SEC}(\mathcal{S})$,
for an endpoint $p$ such that $\overline{pq} \in \mathit{LDS}(\mathcal{S})$,  $|\{ q | \overline{pq} \in \mathit{LDS}(\mathcal{S})\}| \leq 2$.
\end{lemma}
\begin{proof}
Assume that there are three \textit{LDS}s with the endpoint $p$ and let the other endpoints
be $q_1, q_2$ and $q_3$. Since  $q_1, q_2$ and $q_3$ are located on $\mathit{SEC}(\mathcal{S})$,
$\mathit{SEC}(\mathcal{S})$ must be a circle with the center $p$, which is contradiction that $p$ is also
on $\mathit{SEC}(\mathcal{S})$.
\end{proof}

\begin{lemma}\label{lemma:LDSofCH4}
If both $p$ and $q$ are endpoints of \textit{LDS} ($\overline{pq} \in \mathit{LDS}(\mathcal{S}))$ and  are located on $\mathit{SEC}(\mathcal{S})$, 
then $\mathit{Diam}({\cal S})/2 < |\overline{pq}|$.
\end{lemma}
\begin{proof}
Let $c=\mathit{CTR}(\mathcal{S})$ and let $a$ and $b$ be distinct points on $\mathit{SEC}(\mathcal{S})$ such that $\triangle{apc}$ and $\triangle{bpc}$ are equilateral triangles. If $\mathit{Diam}(\mathcal{S})/2 \geq |\overline{pq}|$,  all points in $\mathcal{S}$  must exist in the sector consisting of $\overline{ca}$ and 
$\overline{cb}$ and the shorter arc  $\stackrel{\frown}{a b}$. Then $S\mathit{EC}(\mathcal{S})$ is the circle having $\overline{ab}$ as the diameter. Contradiction.
\end{proof}
%

\begin{lemma}\label{lemma:LDSofCH5}
If $\overline{pq}, \overline{rs} \in \mathit{LDS}(\mathcal{S})$ and $p,q,r,s$ are distinct, then $\overline{pq}$ and  $\overline{rs}$ have an intersection.
\end{lemma}
\begin{proof}
If $\overline{pq}$ and  $\overline{rs}$ do not have any intersection, the longer diagonal of the quadrilateral composed of 
$\overline{pq}$ and $\overline{rs}$ becomes the longest line segment.
\end{proof}

\begin{lemma}\label{lemma:LDSofCH6}
Let $p$ and $q_1 (\overline{pq_1} \in \mathit{LDS}(\mathcal{S}))$ be located on $\mathit{SEC}(\mathcal{S})$
Assume that  $\overline{pq_1}$ is not a diameter of $\mathit{SEC}(\mathcal{S)}$.
Let $q_2$ be the intersection of a circle with center $p$ and radius $|\overline{pq_1}|$ and $\mathit{SEC}(\mathcal{S})$
and not $q_1$. Then there are no points on arc $\stackrel{\frown}{q_1 q_2}$ of  $\mathit{SEC}(\mathcal{S})$ not containing $p$ except $q_1$ and $q_2$ (see Figure~\ref{fig:lemma6}).
\end{lemma}
\begin{proof}
If there is a point $q (\neq q_1, q_2)$ on the arc $\stackrel{\frown}{q_1 q_2}$, $| \overline{pq}|$ is greater than the \textit{LDS}.
\end{proof}

\begin{figure*}
  \begin{center}
    \includegraphics[scale=0.5,width=4cm]{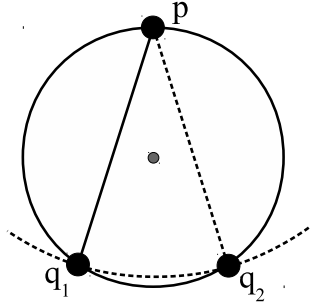}%
    \caption{Proof of Lemma 6.}
    \label{fig:lemma6}
  \end{center}
\end{figure*}

\begin{lemma}\label{lemma:LDSofCH7}
Assume that for any $\overline{pq} \in \mathit{LDS}(\mathcal{S})$, both  $p$ and $q$ are located on $\mathit{SEC}(\mathcal{S})$.
If $\mathcal{H}(\mathcal{S})$ is not a regular polygon, there exists a point $p_0$
such that  $\overline{p_0q} \in \mathit{LDS}(\mathcal{S})$ and
$\overline{p_0r}\not \in \mathit{LDS}(\mathcal{S})$ for any $r (\neq q) \in \mathcal{P} (\mathcal{S})$. The point $p_0$ is called a  {\em single-endpoint} of \textit{LDS} (Figure~\ref{fig:lemma7}(a)).
\end{lemma}
\begin{proof}
If there exists no point $p_0$ such that $\overline{p_0q} \in \mathit{LDS}(\mathcal{S})$ and
$\overline{p_0r}\not \in \mathit{LDS}(\mathcal{S})$ for any $r (\neq q)$, $|\{ q | \overline{pq} \in \mathit{LDS}(\mathcal{S})\}| \geq 2$ for any point $p$ on $\mathit{SEC}(\mathcal{S})$. Since $|\{ q | \overline{pq} \in \mathit{LDS}(\mathcal{S})\}| \leq 2$ by Lemma~\ref{lemma:LDSofCH3},
$|\{ q | \overline{pq} \in \mathit{LDS}(\mathcal{S})\}| = 2$ for any point $p$.
Let $n$ be the number of points located on $\mathit{SEC}(\mathcal{S})$. For a point $p$, let $\overline{pq}$ and $\overline{pr}$ be two \textit{LDS}s
having $p$ as the endpoint. Note that $q$ and $r$ are also located on $\mathit{SEC}(\mathcal{S})$.
Let $pp'$ be the diameter of $\mathit{SEC}(\mathcal{S})$. 
Then $\angle{qpp'}=\angle{rpp'}$ (denoted as $\alpha$). Since this property holds for any point on $\mathit{SEC}(\mathcal{S})$,
$\alpha$ must be  $(n+2)\pi/n$ (Figure~\ref{fig:lemma7}(b)). Otherwise, it is easily verified that it contradict that the number of points located on $\mathit{SEC}(\mathcal{S})$ is $n$. 
When $\alpha=(n+2)\pi/n$, the points on $\mathit{SEC}(\mathcal{S})$ constitute a regular polygon and it contradicts the assumption of this lemma.
\end{proof}

\begin{figure*}
  \begin{center}
    \includegraphics[width=1.0\textwidth, keepaspectratio]{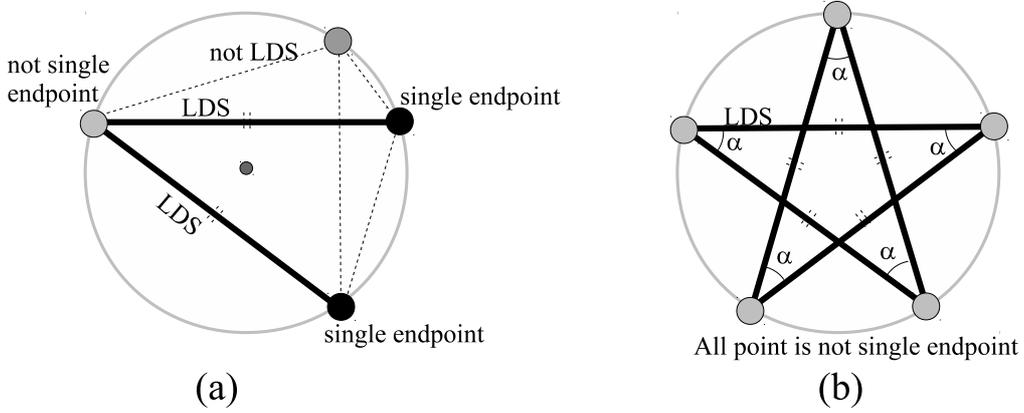}
    \caption{Proof of Lemma 7.}
    \label{fig:lemma7}
  \end{center}
\end{figure*}

\section{Previous Results for Rendezvous and Gathering}
\label{sec:impo}
Rendezvous is trivially solvable with CENT or FSYNC scheduler.
The multiplicity detection does not help to solve Rendezvous and
it is generally unsolved with SSYNC, even if chirality is assumed.

\begin{theorem} \label{theorem:Rand_Impo}
{\em \cite{FPS}} Rendezvous is deterministically unsolvable in SSYNC scheduler even if chirality is assumed.
\end{theorem}

If robots have a constant number of colors in their lights, Rendezvous can be solved shown in the following theorem (see Table~\ref{tab:Table-Rendezvous}).

\begin{theorem}
{\em \cite{HDT,FSVY}}
\begin{enumerate}
\item[(1)] Rendezvous is solved in full-light, non-rigid and ASYNC with $2$ colors.
\item[(2)] Rendezvous is solved in external-light, non-rigid and SSYNC with $3$ colors.
\item[(3)] Rendezvous is solved in external-light, rigid and ASYNC with $12$ colors.
\item[(4)] Rendezvous is solved in external-light, non-rigid and ASYNC with $3$ colors and knowledge of a minimum distance $\delta$ robots move.
\item[(5)] Rendezvous is solved in internal-light, rigid and SSYNC with $6$ colors.
\item[(6)]  Rendezvous is solved in internal-light, non-rigid and SSYNC with $3$ colors and knowledge of a minimum distance $\delta$ robots move.
\end{enumerate} 
\end{theorem}

Impossibility and/or possibility results for Gathering are stated in the following theorems. 

\begin{theorem} \label{theorem:ssync}
{\em \cite{FPS,DGCMR}}
If we do not assume strong multiplicity detection, the followings holds.
\begin{enumerate}
\item[(1)] Gathering is deterministically unsolvable in SSYNC.
\label{theorem:centrr}
\item[(2)] Gathering is deterministically unsolvable in $2$-BOUNDED and CENT.
\item[(3)] Distinct Gathering is deterministically unsolvable in ROUND-ROBIN.
\end{enumerate}
The impossibility holds even if we assume chirality, and rigid movement or non-rigid movement with knowledge of minimum distance $\delta$.
\end{theorem}

\begin{theorem} \label{theorem:GatheringGen}
{\em \cite{FPS}}
\begin{enumerate}
\item[(1)] With strong multiplicity detection, $n$-Gathering is solved in SSYNC if and only if $n$ is odd.
\item[(2)] In ASYNC, with strong multiplicity detection, distinct $n$-Gathering is solved with $n (\geq 3)$ robots.
\end{enumerate}
\end{theorem}

Multiplicity detection is a strong assumption to solve Gathering.
In the following section, without strong multiplicity detection we consider Gathering algorithms for robots with lights such that the number of colors is the minimum and additional assumptions are the weakest.

\section{Gathering Algorithms in SSYNC}
\label{sec:GatheringAlgorithms}

In this section, we show the following four Gathering algorithms with lights in SSYNC.
 
\begin{enumerate}
\item[(a)] {\bf Algorithm~1}:(2, full, SSYNC, non-rigid),
\item[(b)] {\bf Algorithm~2}:(3, external, SSYNC, rigid),
\item[(c)] {\bf Algorithm~3}:(2, external, SSYNC, rigid; local-awareness), and
\item[(d)] {\bf Algorithm~4}+{\bf 5}+{\bf 6}):(2, internal, SSYNC, non-rigid; $2\delta$-distant).
\end{enumerate}

The idea of Gathering algorithms shown here is that the algorithm divides into two steps as follows:

\begin{enumerate}
\item[(1)] We make a configuration where all robots are located on a line segment (called \textit{onLDS}) from any initial configuration satisfying the corresponding conditions.
\item[(2)] We make a gathering algorithm from the configuration \textit{onLDS}.
\end{enumerate}

Step~(1) can be performed with \textit{ElectOneLDS} \cite{IKIW}, that reduces any configuration to one where there is the unique \textit{LDS} or Gathering is achieved. This algorithm needs chirality assumption but can be performed in SSYNC and without lights.
Thus, we will have to use lights to implement Step~(2)
from the configuration \textit{onLDS}. 
It is easily verified that \textit{ElectOneLDS} can be modified so that it can be performed in the following assumptions, (a) non-rigid moving, SSYNC and chirality, (b) non-rigid moving, CENT and non-chirality.

\begin{lemma}\label{lemma:ElectLDS}
With chirality,
an algorithm can be constructed to make \textit{onLDS} from any configuration in SSYNC and non-rigid moving without lights, unless Gathering is achieved. 
If CENT is assumed, this algorithm can be implemented without chirality.
\end{lemma}
\begin{proof}
The algorithm \textit{ElectOneLDS} to make \textit{onLDS} works in rigid movement \cite{FPS}. We show a modified \textit{ElectOneLDS} to work in non-rigid movement. 

The modified \textit{ElectOneLDS} works as follows:
We call a configuration {\em contractible} if (1) its convex hull $CH$ is symmetric and every robot is located at a vertex of \textit{CH} or at the center of \textit{CH}, or (2) \textit{CH} is not symmetric and there are no robots inside \textit{CH}.
The algorithm makes a contractible configuration from the current configuration, unless it already contractible. It can be done by each active robot that
is not at a vertex of \textit{CH} moves to the center of \textit{CH} if it is symmetric,
or each active robot that 
is neither at a vertex of \textit{CH} nor on an edge of \textit{CH} (that is, inside of \textit{CH}) moves to a vertex of \textit{CH}. 
In making contractible configurations, since the convex hull is not changed,
the algorithm can work in non-rigid moving.
Note that robots may be located on edges of \textit{CH} in contractible and  non-symmetric configurations but can not be located on edges of \textit{CH} in contractible and symmetric ones.

If \textit{CH} becomes contractible, a unique \textit{LDS} is obtained by decreasing the number of edges of \textit{CH} or the diameter of \textit{CH} until the configuration has a unique \textit{LDS}.
If \textit{CH} is contractible and symmetric, any active robot moves to the center of \textit{CH}.
In this case, if \textit{CH} is not changed the number of robots at vertices of \textit{CH} decreases
via non-contractible configurations. Otherwise, the configuration becomes one of the following ones unless Gathering is attained;
 (1) it remains contractible and symmetric,  (2) the number of edges of the convex hull is decreased and it becomes contractible one via  non-contractible configurations. In the both cases, the diameter of \textit{CH} is eventually reduced.
 
 If \textit{CH} is contractible and asymmetric, the number of edges is decreased by some robots moving so to contract shortest-length edges in \textit{CH}. Each robot $r$ on a vertex of \textit{CH} checks the distance to the clockwise and  counter-clockwise neighboring robots at vertices of \textit{CH}.
If the edge on the left is a shortest-length edge and the robot is located at the leftmost,
this edge becomes contracting. The leftmost robot and the robots on this edge move to the 
the other vertex of the edge until the edge is contracting. Since all non-contracting robots do not move and the edge remains the shortest-length, this contraction can be done even in non-rigid moving, In this way, the number of edges of the convex hull is decreasing. 

In SSYNC, chirality is necessary for determining the leftmost robot of the shortest-length edge. However, since only one robot is activated in CENT, contracting the shortest-length edge can be done without chirality.  
\end{proof}

\subsection{Full-light}

\Newcodeline
\begin{algorithm}[h]
\caption{Full-Light-Garher($r_i$)}
\label{algo:FLG}
{\footnotesize
\begin{tabbing}
111 \= 11 \= 11 \= 11 \= 11 \= 11 \= 11 \= \kill
{\em Assumptions}: full-light, non-rigid, $2$ colors($A$ and $B$), set-view, SSYNC\crm
{\em Input}: configuration \textit{onLDS}, all robots have color $A$ \crm
\Cl \> {\bf case} ${\cal L}({\cal SS}_i)$  {\bf of } \crm
\Cl \> $\{ A \}$: \crm
\Cl \> \> {\bf if} $|{\cal P}_A({\cal SS}_i)| = 1$ {\bf then} $des_i \leftarrow p_i$ // gathered!\crm
\Cl \> \> {\bf else if} $|{\cal P}_A({\cal SS}_i)|  = 2$ {\bf then}  $\ell_i \leftarrow B$; $des_i \leftarrow (p_n + p_f)/2$\crm
\Cl \> \> {\bf else} //$|{\cal P}_A({\cal SS}_i)|  \geq 3$\crm
\Cl \> \> \> {\bf if} $(p_i = p_n)$ {\bf then} $des_i \leftarrow p_i$ //I am at either endpoint\crm 
\Cl \> \> \> {\bf else} $des_i \leftarrow p_n$\crm
\Cl \> $\{B \}$: \crm
\Cl \>\> {\bf if} $(p_i = p_n)$ {\bf then} $\ell_i \leftarrow A$\crm
\Cl \> \>  $des_i \leftarrow p_i$ // stay\crm
\Cl \> $\{A, B \}$: \crm
\Cl \> \> {\bf if} $\ell_i = A$ {\bf then}\crm
\Cl \> \> \>{\bf if} $|{\cal P}_A({\cal SS}_i)|  = 1$ {\bf then} $des_i \leftarrow p_i$ // stay\crm
\Cl \> \>  \>{\bf else if} ${\cal P}_A({\cal SS}_i) = \{ p_n,p_f\}$ {\bf and}  ${\cal P}_B({\cal SS}_i)  = \{ \frac{p_n+p_f}{2} \}$ {\bf then}\crm
\Cl \> \> \>\> $\ell_i \leftarrow B$ \crm
\Cl \> \> \> \>$des_i \leftarrow  (p_n + p_f)/2$\crm
\Cl \> \> {\bf else} // $\ell_i = B$\crm 
\Cl \> \> \>{\bf if} $|{\cal P}_A({\cal SS}_i)|  = 1$ {\bf then} $des_i \leftarrow p (\in {\cal P}_A({\cal SS}_i)) $\crm
\Cl \> \> \>{\bf else} $des_i \leftarrow (p_n + p_f)/2$ //$|{\cal P}_A({\cal SS}_i)|  = 2$ \crm 
\Cl \> {\bf endcase} 
\end{tabbing}
}
\end{algorithm}

Gathering algorithm with full-light uses two colors of full-light and works in non-rigid movement shown in  {\bf Algorithm~\ref{algo:FLG}}.
It reduces any configuration of \textit{onLDS} to a gathering configuration. 
In the algorithm for robot $r_i$, endpoints of the \textit{LDS} are denoted by $p_n$ and $p_f$, where $p_n$ is the nearest endpoint from $p_i$ and $p_f$ is the farthest one from $p_i$. Note that if $p_i$ is either of endpoints, $p_i=p_n$ and $p_f$ is the other endpoint.
Since full-light is assumed, robot $r_i$ can use %
colors of other robots in $\mathcal{SS}_i$  and its own light $\ell_i$ in the algorithm.
For snapshot $\mathcal{SS}_i$ and color $A$ and $B$,  $\mathcal{P}_A(\mathcal{SS}_i)$ and $\mathcal{P}_B(\mathcal{SS}_i)$ denote sets of positions on which robots with color $A$ and $B$ are located, respectively.

\begin{figure*}
  \begin{center}
    \includegraphics[scale=0.8]{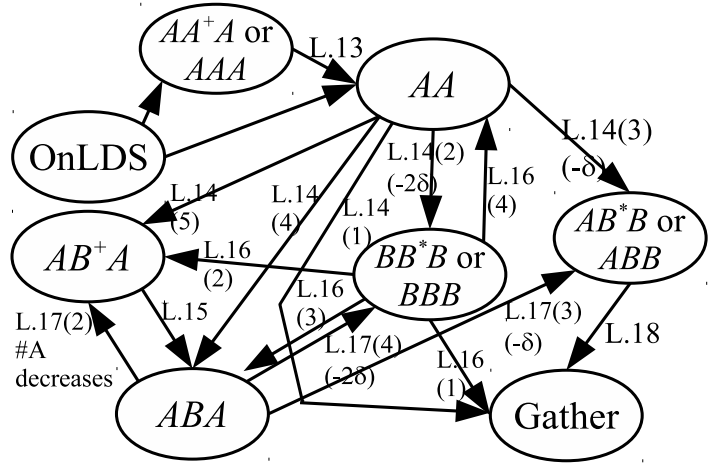}
    \caption{Transition diagram for Full-Light-Gather.}
    \label{fig:Full-light}
  \end{center}
\end{figure*}

Initially all robots are located  on the unique \textit{LDS} and have color A.
Figure~\ref{fig:Full-light} shows the transition diagram of
{\bf Algorithm~\ref{algo:FLG}}, where nodes denote configurations and a directed edge denotes transition from a configuration to a configuration. In Figure~\ref{fig:Full-light}, labels of directed edges have the following meanings.
L.\# means that this transition is shown by Lemma~\#.
``$-\delta$'' (``$2\delta$'') means when the transition occurs, the distance between the endpoints reduces at least $-\delta (2\delta)$ and ``\#A decreases'' means the numbers of robots with color A at the endpoints decrease. 
Each configuration is denoted by a regular-expression-like sequence of colors robots have from one endpoint to another. 
Formally we define {\em color-configurations} as follows. 
Let $C(t)$ be a configuration at time $t$,  $p$ and $q$ be the endpoints of the \textit{LDS}.
Configuration $C(t)$ has a {\em color-configuration} 
\begin{enumerate}
\item[(1)] $\alpha  \beta$, 
if all robots at $p$ have color $\alpha$, all robots at $q$ have color $\beta$ ($\alpha, \beta \in \{ A, B\}$) and
there are no robots inside the \textit{LDS},
\item[(2)] $\alpha \gamma \beta$, 
if all robots at $p$ have color $\alpha$, all robots at $q$ have color $\beta$ ($\alpha, \beta \in \{ A, B\}$), 
all robots at the mid-point of the \textit{LDS} have color $\gamma$ and there are no robots except on the three locations, and
\item[(3)] $\alpha \gamma^+ \beta$, 
if all robots at $p$ have color $\alpha$, all robots at $q$ have color $\beta$ ($\alpha, \beta \in \{ A, B\}$) and 
there is at least one location except the midpoint of the \textit{LDS} inside the \textit{LDS} where all robots have color $\gamma$.
\end{enumerate}
Note that $\alpha \gamma^+ \beta$ contains configurations that the midpoint of the \textit{LDS} may be occupied by some robot(s) with color $\gamma$, and 
$\alpha \gamma \beta$ and $\alpha \gamma^+ \beta$ are exclusive.
If either (1) or (3) is satisfied, we denote $\alpha \gamma^* \beta$.
Let $dis(C(t))$ denote the length of the \textit{LDS} in the configuration $C(t)$.

The outline of behaviour of {Algorithm~\ref{algo:FLG}} is explained as follows.
Suppose that all robots become active (FSYNC).
From any initial color-configuration of \textit{onLDS} ($AA^*A$ or $AAA$), robots with $A$ located not at endpoints move to endpoints and robots with $A$ located at endpoints stay({\bf lines}~5-7) it becomes the color-configuration $AA$.  
When the color-configuration is $AA$, each robot changes its color to $B$ and move to the midpoint({\bf line}~4).  
If they can reach to the destination, the gathering is achieved.
However, since we assume non-rigid movement, some of them may stop before reaching the midpoint and its configuration becomes 
$BB^*B$ or $BBB$, where the former is $BB$ if the locations all robots stop are only two, otherwise $BB^+B$ and 
the latter is possible when there is a robot moving to the midpoint.
In the both cases the length of its \textit{LDS} is decreased by at least $2\delta$ (denoted by $-2\delta$ on the directed edge from $AA$ to $BB^*B$ in Fig.~\ref{fig:Full-light}).

From the color-configuration $BB^*B$ or $BBB$, 
the algorithm changes colors of robots at endpoints to $A$ ({\bf lines}~8-10) and can change the color-configuration directly 
to $AB^+A$ (from $BB^+B$), $AA$ (from $BB$), or $ABA$ (from $BBB$).
If the color-configuration $AB^+A$ occurs,
robots with $B$ move to the midpoint ({\bf line}~19),  and robots with $A$ stay until all robots with $B$ move to the midpoint({\bf lines}~14-16).
In this case the configuration also become $ABA$ and robots with $A$ change the color to $B$ and move to the midpoint ({\bf lines}~14-16) and 
Gathering  is achieved after repeating the transitions from $ABA$ to $BB^*B$, from $BB^*B$ to $AB^+A$, and from $AB^+A$ to $ABA$.
If the color-configuration $AA$ occurs, repeating the above transitions Gathering attained.

However, since all robots do not become active at every round in general,
the behaviors are complicated and all color-configurations shown in Figure~\ref{fig:Full-light} can occur.

Since any cycle in the graph can reduce the length of \textit{LDS} (denoted by $-\delta$ or $-2\delta$ in Fig.~\ref{fig:Full-light}) or the number of robots with $A$ and located at endpoints (denoted by $\#A$ decreases in Figure~\ref{fig:Full-light}), we can show that {\bf Algorithm~\ref{algo:FLG}} solves Gathering from any initial configuration of \textit{onLDS}.

\begin{lemma} \label{AAA-AA}
If $C(t)$ is a configuration at time $t$ with color-configuration $AA^*A$ or $AAA$, there is a time $t' (\geq t)$ such that
$C(t) \rightarrow^* C(t')$, $C(t')$ has a color-configuration $AA$, and $dis(C(t'))=dis(C(t))$. 
\end{lemma}
\begin{proof}
If the color-configuration is $AA$, the lemma holds as $t'=t$.
Otherwise, since robots located inside the \textit{LDS} move to the nearest endpoint and robots at the endpoint stay there as long as there is a robot
located inside the \textit{LDS} ({\bf lines}~5-7), the color-configuration becomes $AA$.
\end{proof}

\begin{lemma}
Let $C(t)$ be a configuration at time $t$ with color-configuration $AA$.
If we consider time $t+1$ such that $C(t) \rightarrow C(t+1)$,  then 
\begin{enumerate}
\item[(1)] $C(t+1)$ is a gathering configuration,
\item[(2)] $C(t+1)$ has a color-configuration $BB^*B$ or $BBB$, and $dis(C(t+1)) \leq dis(C(t)) - 2 \delta$,
\item[(3)] $C(t+1)$ has a color-configuration $AB^*B$ or $ABB$ (or $BB^*A$ or $BBA$), and $dis(C(t+1)) \leq dis(C(t)) - \delta$,
\item[(4)] $C(t+1)$ has a color-configuration $ABA$, and $dis(C(t+1))=dis(C(t))$, or
\item[(5)] $C(t+1)$ has a color-configuration $AB^+A$, and $dis(C(t+1))=dis(C(t))$.
\end{enumerate}
\end{lemma}
\begin{proof}
There are five cases according to the activeness of robots.
\begin{enumerate}
\item[(1)] If all robots become active and they can reach the destination (the midpoint of the \textit{LDS}), the gathering is attained.
\item[(2)] If all robots become active but they can not reach the destination, they change their color to $B$ and move at least $\delta$.
Thus the color-configuration becomes $BB^+B$, $BB$ or $BBB$ as shown in the example above.
In any case, since all robots move at least $\delta$, $dis(C(t+1)) \leq dis(C(t)) - 2 \delta$.
\item[(3)] If all robots at one endpoint become active but some robots at the other endpoint are inactive,  all active robots change their color to $B$ and all inactive robots
do not change their color and stay there. Since the color of all robots at the endpoint inactive robots are located is $A$ and the color of the other robots become $B$ they move at least $\delta$, the color-configuration becomes $AB^*B$ $ABB$ (or $BB^*A$ or $BBA$). In this case since all active robots move at least $\delta$,
$dis(C(t+1)) \leq dis(C(t)) - \delta$.
\item[(4)] If there exist inactive robots at the both endpoints and active robots can reach the destination, the color-configuration becomes $ABA$ and $dis(C(t+1))
=dis(C(t))$.
\item[(5)] If there exist inactive robots at the both endpoints and some active robots can not reach the destination, the color-configuration becomes $AB^+A$ and $dis(C(t+1))=dis(C(t))$.
\end{enumerate}
\end{proof}

\begin{lemma}
If $C(t)$ is a configuration at time $t$ with color-configuration $AB^+A$, there is a time $t' (\geq t)$ such that
$C(t) \rightarrow^* C(t')$, $C(t')$ has a color-configuration $ABA$, and $dis(C(t'))=dis(C(t))$. 
\end{lemma}
\begin{proof}
In the configuration with color-configuration $AB^+A$,
robots with $A$ do not move until the color-configuration becomes  $ABA$ ({\bf lines}~12-16) 
robots with $B$ move to the midpoint of the \textit{LDS}. Thus it becomes $ABA$ and $dis(C(t'))=dis(C(t))$.
\end{proof}

\begin{lemma}
If $C(t)$ is a configuration at time $t$ with color-configuration $BB^*B$ or $BBB$, there is a time $t' (> t)$ such that
$C(t) \rightarrow^* C(t')$ and one of the followings holds,
\begin{enumerate}
\item[(1)] $C(t')$ is a gathering configuration,
\item[(2)] 
$C(t')$ has a color-configuration $AB^+A$, and $dis(C(t')) = dis(C(t))$,
\item[(3)] 
$C(t')$ has a color-configuration $ABA$, and $dis(C(t'))=dis(C(t))$, or
\item[(4)] 
$C(t')$ has a color-configuration $AA$, and $dis(C(t'))=dis(C(t))$.
\end{enumerate}
\end{lemma}
\begin{proof}
Assume that all robots become active at time $t$.
If the color-configuration is $BB$, then the color-configuration becomes $AA$ at time $t'=t+1$  (case (4)).
If the color-configuration is $BB^+B$ or $BBB$, then the color-configuration becomes $AB^+A$ or $ABA$ at time $t'=t+1$, respectively (case (2) or (3)). These transitions occur from {\bf lines}~8-10 of the algorithm.

Otherwise, there are two cases, (a) some robots at only one endpoint become active, and (b) robots at the both endpoints become active. Note that active robots not located on the endpoints do nothing in these configurations.
\begin{enumerate}
\item[(a)] In this case, active robots change their color into $A$ at time $t+1$ and since $|\mathcal{P}_A(C(t+1))|  = 1$,
robots with $A$ stay at the positions ({\bf lines}~12-13)  and robots with $B$ move to the position at robots with $A$ ({\bf lines}~17-18) 
and the gathering is attained (case (1)).
\item[(b)] In this case, active robots at the both endpoints change their colors into $A$ at time $t+1$ and since $|{\cal P}_A(C(t+1))|  = 2$ and
robots with $A$ do not move and robots with $B$ move to the midpoints of the \textit{LDS}, the color-configuration eventually becomes $ABA$ (case (3)).
\end{enumerate}
\end{proof}

\begin{lemma}
Let $C(t)$ be a configuration at time $t$ with color-configuration $ABA$. If we consider time $t'$ such that $C(t) \rightarrow^* C(t')$, then
\begin{enumerate}
\item[(1)] $C(t')$ is a gathering configuration,
\item[(2)] $C(t')$ has a color-configuration $AB^*A$, and $dis(C(t'))=dis(C(t))$ and the number of robots
at the endpoints and with $A$ is decreased,
\item[(3)] $C(t')$ has a color-configuration $AB^*B$ (or $BB^*A$), and $dis(C(t')) \leq dis(C(t)) - \delta$, or 
\item[(4)] $C(t')$ has a color-configuration $BB^*B$ or $BBB$, and $dis(C(t')) \leq dis(C(t))-2\delta$.
\end{enumerate}
\end{lemma}
\begin{proof}
This lemma can be proved similar to the one of Lemma~\ref{AAA-AA}.
\end{proof}

\begin{lemma}\label{lemma:ABBG}
If $C(t)$ is a configuration at time $t$ with color-configuration $AB^*B$ or $ABB$ (or $BB^*A$ or $BBA$), there is a time $t' (\geq t)$ such that
$C(t) \rightarrow^* C(t')$, $C(t')$ is a gathering configuration. 
\end{lemma}
\begin{proof}
Since $|{\cal P}_A(C(t))|  = 1$,
robots with $A$ stay at the position ({\bf lines}~12-13)  and robots with $B$ move to the position at robots with $A$ ({\bf lines}~17-18) 
and the gathering is attained.
\end{proof}

We have the following theorem using Lemmas~\ref{lemma:ElectLDS}-\ref{lemma:ABBG}.

\begin{theorem}
\label{th:Full-light}
Gathering is solvable in full-light of $2$ colors, non-rigid, and SSYNC, with set-view and agreement of chirality.
\end{theorem}

\subsection{External-light}

\Newcodeline
\begin{algorithm}[h]
\caption{Ext-Light-Gather-with-3colors($r_i$)}
\label{algo:ELGW3C}
{\footnotesize
\begin{tabbing}
111 \= 11 \= 11 \= 11 \= 11 \= 11 \= 11 \= \kill
{\em Assumptions}: external-light, rigid, $3$ colors ($A$,$B$,$C$), local-unawareness, set-view, SSYNC\crm
{\em Input}: $(|{\cal P}({\cal SS}_i)| =3)$ {\bf and} $(|\overline{p_f p_m}|=|\overline{p_m p_n}|)$ or $|{\cal P}({\cal SS}_i)| =2$,  all robots have color $A$.\crm
\Cl \> {\bf if} $|{\cal P}({\cal SS}_i)| =2$ {\bf then} \crm %
\Cl \> \> {\bf if} $C \in {\cal L}_{p_f}({\cal SS}_i)$ {\bf then} $\ell_i \leftarrow C$; $des_i \leftarrow p_f$ \crm
\Cl \> \> {\bf else  if} ${\cal L}_{p_f}({\cal SS}_i) = \{ A \}$ {\bf then} $\ell_i \leftarrow B$; $des_i \leftarrow (p_n + p_f)/2$\crm
\Cl \> \> {\bf else  if} ${\cal L}_{p_f}({\cal SS}_i) = \{ B \}$ {\bf then} $\ell_i \leftarrow C$; $des_i \leftarrow p_i$  // stay\crm
\Cl \> {\bf else if} $(|{\cal P}({\cal SS}_i)| =3)$ {\bf and} $(|\overline{p_f p_m}|=|\overline{p_m p_n}|)$ {\bf  then}\crm
\Cl \> \> \> {\bf if} $p_i =p_m$  {\bf then} //I am at the midpoint.\crm
\Cl \> \> \> \>$\ell_i \leftarrow B$\crm
\Cl \> \> \> \> $des_i \leftarrow p_m$ //stay\crm
\Cl \> \> \> {\bf else if} $ {\cal L}_{p_f}({\cal SS}_i)=\{B\}$ {\bf and} $ {\cal L}_{p_m}({\cal SS}_i)=\{B\}$ {\bf then}\crm
\Cl \> \> \> \>\>$\ell_i \leftarrow C$\crm
\Cl \> \> \> \> \>$des_i \leftarrow p_i$ //stay\crm
\Cl \> \> \> {\bf else if} $ {\cal L}_{p_f}({\cal SS}_i)\cap \{A, C\} \not = \emptyset$ {\bf and} $ {\cal L}_{p_m}({\cal SS}_i)=\{B\}$ {\bf then}\crm
\Cl \> \> \> \>\>$\ell_i \leftarrow B$\crm
\Cl \> \> \> \> \>$des_i \leftarrow p_m$ //go to the midpoint
\end{tabbing}
}
\end{algorithm}

We show two algorithms in external-light. One uses $3$ colors in rigid and SSYNC and the other uses only $2$ colors in the same assumption as the former
with the local-awareness. 
{\bf Algorithm~\ref{algo:ELGW3C}}  shows the external-light algorithm in rigid and SSYNC with $3$ colors
and its transition diagram is shown in Figure~\ref{fig:Ext-light}.

\begin{figure*}
  \begin{center}
    \includegraphics[scale=0.8,width=8cm]{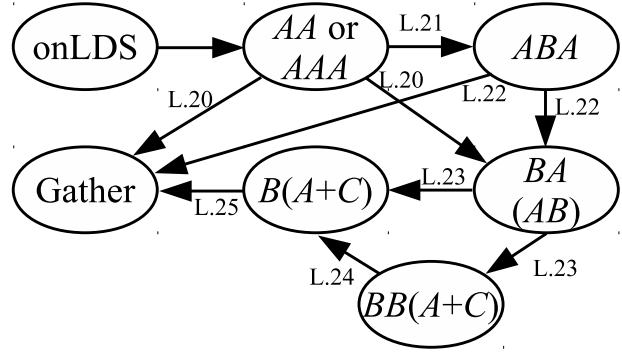}%
    \caption{Transition diagram for Ext-Light-Gather-with-3colors.}
    \label{fig:Ext-light}
  \end{center}
\end{figure*}

Without loss of generality, the algorithm may begin with a configuration of two-location $p_n$ and $p_f$, where all robots have $A$ or three-location $p_n, p_m$, and $p_f$
such that $(|\overline{p_f p_m}|=|\overline{p_m p_n}|)$, where all robots have $A$. Since we assume the rigid movement,
these configurations can be reduced from \textit{onLDS} by moving robots to the nearest position among the endpoints and the midpoint. 
By using notations of the color-configurations, the former has color-configuration $AA$ and the latter has $AAA$ in Fig.~\ref{fig:Ext-light}.  
Once the configuration of two- or three-location is obtained, this configuration is preserved hereafter 
since we assume the rigid movement.

{\bf Algorithm~\ref{algo:ELGW3C}} makes configuration with $BA$ (or $AB$) from the initial configuration with $AAA$ or $AA$
via configuration with $ABA$.  After the configuration with $BA$ \footnote{$AB$ can be treated similarly}, (1) when robots observe color $B$, they change their color to $C$ and stay there, where the location becomes the gathering one, (2) when robots observe color $A$, they change their color to $B$ and go to the midpoint. Then, the color-configuration becomes $B(A+C)$ or $BB(A+C)$, where $A+C$ denotes the endpoint is occupied by robots with color $A$ or $C$.
Once the configuration with $B(A+C)$ appears, since robots with $A$ or $C$ observe color $B$ and change their color to $C$, and robots with $B$ observe color $C$ and move to the location having robots with $C$, Gathering is successful.
The configuration with $BB(A+C)$ becomes one with $B(A+C)$, because robots with $B$ located at the midpoint stay there, 
robots with $A$ or $C$ observe two $B$'s at the endpoint and the midpoint change their color to $C$, robots with $B$ located at the endpoint observe color $C$ at the endpoint and color $B$ at the midpoint move to the midpoint.

\begin{lemma}\label{AAtoABA}
Configuration with $AA$ moves to one with $ABA$, $AB$ or $BA$. 
In particular, it moves to the gathering configuration,
if all robots become active.  
\end{lemma}
\begin{proof}
All active robots perform {\bf line}~3. 
If all robots become active, the gathering is attained.
Otherwise, if all robots at one endpoint only become active, the configuration moves to $AB$ or $BA$.
Otherwise, it moves to $ABA$.
\end{proof}

\begin{lemma}\label{AAAtoABA}
Configuration with $AAA$ moves to one with $ABA$.   
\end{lemma}
\begin{proof}
This case is derived from  {\bf line}~3,   {\bf lines}~6-8, and the fact that robots at endpoints cannot move until all robots 
at the midpoint have color $B$ ({\bf line}~12).
\end{proof}

\begin{lemma}\label{ABAtoBA}
Configuration with $ABA$ moves to the two-point configuration with $BA$ (or $AB$) unless Gathering is achieved.
\end{lemma}
\begin{proof}
From a configuration with $BA$,
robots at the endpoints change their color to $B$ and go to the midpoint ({\bf lines}~12-14), and robots at the midpoint
stay there ({\bf lines}~6-8). Then if all robots at the both endpoints go to the midpoint, Gathering is achieved.
Otherwise, the two-point configuration
with $BA$ (or $AB$) is obtained. 
\end{proof}

\begin{lemma}\label{BAtoBA+C}
Configuration with $BA$ moves to one with $B(A+C)$ or $BB(A+C)$.
\end{lemma}
\begin{proof}
Robots with $A$ observing $B$ change their color to $C$ and stay ({\bf line}~4) and robots with $B$ observing $A$ change their color to $B$ and move to the midpoint ({\bf line}~3). Thus, the color-configuration $B(A+C)$ or $BB(A+C)$ is obtained. 
\end{proof}

\begin{lemma}
Configuration with $BB(A+C)$ moves to one with $B(A+C)$.
\end{lemma}
\begin{proof}
It can be derived by {\bf lines}~5-14.
\end{proof}

\begin{lemma}\label{BCtoGa}
Configuration with $B(A+C)$ moves to the gathering configuration.
\end{lemma}
\begin{proof}
It can be derived by {\bf lines}~1-4.
\end{proof}

We have the following theorem using Lemma~\ref{lemma:ElectLDS} and Lemmas~\ref{AAtoABA}-\ref{BCtoGa}. 

\begin{theorem}
Gathering is solvable in external-light of $3$ colors, rigid, and SSYNC, if robots have set-view and agreement of chirality.
\end{theorem}

\Newcodeline
\begin{algorithm}[h]
\caption{Ext-Light-Gather-with-2-Colors($r_i$)}
\label{algo:ELG2C}
{\footnotesize 
\begin{tabbing}
111 \= 11 \= 11 \= 11 \= 11 \= 11 \= 11 \= \kill
{\em Assumptions}: external-light, rigid, 2 colors (A,B), local-awareness, set-view, SSYNC. \crm
{\em Input}: $(|{\cal P}({\cal SS}_i)| =3)$ {\bf and} $(|\overline{p_f p_m}|=|\overline{p_m p_n}|)$ or $|{\cal P}({\cal SS}_i)| =2$, all robots have color $A$. \crm
\Cl \> {\bf if} $|{\cal P}({\cal SS}_i)| =2$ {\bf then} \crm 
\Cl \> \> {\bf if} $B \in {\cal L}({\cal SS}_i)$ {\bf then} $\ell_i \leftarrow B$; $des_i \leftarrow p (\in  {\cal P}_B({\cal SS}_i))$\crm
\Cl \> \> {\bf else  if} ${\cal L}_{p_i}({\cal SS}_i) = \emptyset$ {\bf then} no action //I am alone at $p_i$\crm
\Cl \> \> {\bf else} $\ell_i \leftarrow B$; $des_i \leftarrow (p_n + p_f)/2$ //${\cal L}_{p_f}({\cal SS}_i) = \{ A \}$\crm %
\Cl \> {\bf else if} $(|{\cal P}({\cal SS}_i)| =3)$ {\bf and} $(|\overline{p_f p_m}|=|\overline{p_m p_n}|)${\bf then}\crm
\Cl \> \> \> {\bf if} $A \in {\cal L}_{p_m}({\cal SS}_i)$ {\bf and} ($p_i =p_m$) {\bf then} $\ell_i \leftarrow B$\crm
\Cl \> \> \>{\bf else }  $\ell_i \leftarrow B$; $des_i \leftarrow p_m$
\end{tabbing}
}
\end{algorithm}

Although it is unknown whether there exists a $2$-color gathering algorithm in external-light, rigid, and SSYNC,
we can construct $2$-color gathering algorithm in the same condition if we assume the local-awareness ({\bf Algorithm~\ref{algo:ELG2C}}). In {\bf Algorithm~\ref{algo:ELGW3C}}, we make the location having robots with $C$ the gathering point.
Instead we do not use the third color $C$, we make configurations with $AB$ (or $BA$) such that
there are at least 2 robots with $B$ at one endpoint utilizing the local-awareness.
Thus, since robots with $A$ move to the point having robots with $B$ and robots with $B$ can stay at the point({\bf line}~2) using the local-awareness, Gathering is attained.
Such configurations are created as follows; 
When the initial color-configuration is $AAA$,
it moves to $ABA$ ({\bf line}~6) and reaches to $BA$ (or $AB$) ({\bf lines}~7).
When one endpoint contains more than one robot in the initial configuration with $AA$,
it moves to a configuration with $ABA$ ({\bf line}~4) and reaches to one with $BA$ (or $AB$) ({\bf lines}~5-7).
If one endpoint contains only one robot in configuration with $AA$,
since the robot identifies that it is alone due to the local-awareness, the robot makes no action. 
Thus, all the initial configurations can reach to the desired one with $BA$ (or $AB$).

\begin{theorem}
Assume the local-awareness. 
Gathering is solvable in external-light of $2$ colors, rigid, and SSYNC, if robots have set-view and agreement of chirality.
\end{theorem}

\subsection{Internal-light}

Views of robots in internal-light are the same as those without lights and so robots must determine their behaviors by these views without colors and their own colors of lights. Thus  Gathering algorithms in internal-light can not seem to be constructed  without additional knowledge such as distance information.
In fact, known Rendezvous algorithms use the minimum distance of moving $\delta$  and/or  the unit distance \cite{FSVY}.

In our Gathering algorithm in internal-light, 
we use the minimum distance of moving $\delta$ and
assume initial configurations such that
if $r_i$ and $r_j$ do not occupy the same location,
$dis(p_i,p_j) \geq D$, where $D$ is a function of $\delta$. 
This condition is said to be {\em $D$-distant}\footnote{This assumption does not imply distinct Gathering.}. 
If any initial configuration is $D$-distant such that $D \leq 2\delta$,  
we can construct a gathering algorithm in internal-light, non-rigid with knowledge of $\delta$, and SSYNC with $2$ colors of lights.

Our Gathering algorithm with internal-light is composed of  three parts.
\begin{enumerate}
\item[(1)] From any $D$-distant configuration, we make a $2$-point configuration with the distance is at least $D/2$ ({\bf Algorithm~\ref{algo:ELDSP}}).
\item[(2)] From any $2$-point configuration with distance $d \geq D/2$, we make a $2$-point configuration with $D/2 > d \geq D/4$ ({\bf Algorithm~\ref{algo:RDLDS}}).
\item[(3)] From any $2$-point configuration with $D/2 > d \geq D/4$, we make a Gathering configuration ({\bf Algorithm~\ref{algo:ILG}}).
\end{enumerate} 

We do not use colors to solve (1) and (2) and only use two colors to solve (3).
The output of (i) is the input of (i+1) (i=1,2) and
the explanation of algorithms is stated in reverse order from (3) to (1).

{\bf Algorithm~\ref{algo:ILG}} solves Gathering in internal-light if initial configurations satisfy 
$|{\cal P}({\cal SS}_i)| =2$ and $\frac{D}{4} \leq |\overline{p_np_f}| <\frac{D}{2}$. 
Since $D \leq 2\delta$, every movement in {\bf Algorithm~\ref{algo:ILG}} is the same as the rigid one.
The following lemma is easily verified for {\bf Algorithm~\ref{algo:ILG}}.

\begin{lemma}\label{Solve-1}
Let $\mathcal{C}(t)$ be a configuration at time $t$ appearing in {\bf Algorithm~\ref{algo:ILG}}.
\begin{enumerate}
\item[(1)] If $\mathcal{C}(t)$ satisfies $|\mathcal{P}(\mathcal{C}(t))| =2$ and $\frac{D}{4} \leq \mathit{Diam}(\mathcal{C}(t)) <\frac{D}{2}$ and has a color-configuration $AA$,
then (1-I) $\mathcal{C}(t+1)$ is a Gathering configuration, (1-II) $\mathcal{C}(t+1)$ has a color-configuration $AB$ or $BA$ and
$\frac{D}{8} \leq \mathit{Diam}(\mathcal{C}(t+1)) <\frac{D}{4}$, or
(1-III) $C(t+1)$ has a color-configuration $ABA$ and
$\frac{D}{4} \leq \mathit{Diam}(\mathcal{C}(t+1)) <\frac{D}{2}$.
\item[(2)] If $\mathcal{C}(t)$ has a color-configuration $ABA$ and
$\frac{D}{4} \leq \mathit{Diam}(\mathcal{C}(t)) <\frac{D}{2}$, then there is a time $t' (t' >t)$ such that $\mathcal{C}(t) \rightarrow^* \mathcal{C}(t')$ and (2-I) $\mathcal{C}(t')$ is a Gathering configuration, or $\mathcal{C}(t')$ has a color-configuration $BA$ or $AB$ and $\frac{D}{8} \leq \mathit{Diam}(\mathcal{C}(t')) <\frac{D}{4}$.
\item[(3)] If $\mathcal{C}(t)$ has a color-configuration $BA$ or $AB$ and $\frac{D}{8} \leq \mathit{Diam}(\mathcal{C}(t)) <\frac{D}{4}$, then there is a time $t' (t' >t)$ such that $\mathcal{C}(t) \rightarrow^* \mathcal{C}(t')$ and $\mathcal{C}(t')$ is a Gathering configuration.
\end{enumerate}
\end{lemma}

\Newcodeline
\begin{algorithm}[h]
\caption{Int-Light-Gather($r_i$)}
\label{algo:ILG}
{\footnotesize
\begin{tabbing}
111 \= 11 \= 11 \= 11 \= 11 \= 11 \= 11 \= \kill
{\em Assumptions}: internal-light, 2 colors (A,B), non-rigid with $\delta$, $D \leq 2\delta$, set-view, SSYNC. \crm
{\em Input}:  $|\mathcal{P}(\mathcal{SS}_i)| =2$ and $\frac{D}{4} \leq |\overline{p_np_f}| <\frac{D}{2}$ , all robots have color $A$.\crm
\crm
\Cl \> {\bf if} $\frac{D}{4} \leq |\overline{p_np_f}| <\frac{D}{2}$ {\bf then} \crm %
\Cl \>\> {\bf if} $(|\mathcal{P}(\mathcal{SS}_i)| =2)$ {\bf or}  $((|\mathcal{P}(\mathcal{SS}_i)| =3)$ {\bf and} $(|\overline{p_f p_m}|=|\overline{p_m p_n}|))${\bf then} \crm %
\Cl \> \>\> {\bf if} $\ell_i = A$ {\bf then}\crm
\Cl \> \> \>\> $\ell_i \leftarrow B$ \crm
\Cl \> \> \> \>$des_i \leftarrow p_m$ \crm
\Cl \> {\bf else} //$\frac{D}{4} > |\overline{p_np_f}| (\geq \frac{D}{8})$ \crm
\Cl \> \> {\bf if} $\ell_i =A$ {\bf then}\crm
\Cl \> \> \>$des_i \leftarrow  p_f$ \crm %
\Cl \>\> {\bf else} //$\ell_i=B$ \crm
\Cl \>\> \>$des_i \leftarrow  p_i$// stay
\end{tabbing}
}
\end{algorithm}

Next we make 2-point location with distance $d$ satisfying $D/2 > d \geq D/4$ from any 2-point locations with $d \geq D/2$. 
This adjusting task is done by {\bf Algorithm~\ref{algo:RDLDS}}. 

In {\bf Algorithm~\ref{algo:RDLDS}}, \textit{A3P} and \textit{A4P} denote the following predicates (Figure~\ref{fig:A3PA4P}). If robots on the both endpoints move, its configuration satisfies \textit{A3P} and if robots on only one endpoint move, its configuration satisfies .\textit{A4P}  

\begin{enumerate}
\item[\textit{A3P}] : $|\mathcal{P}(\mathcal{SS}_i)| =3$, and letting  $\mathcal{P}(\mathcal{SS}_i) =\{ p_n, p_m, p_f\}$,  $|\overline{p_np_{m}}|  \neq |\overline{p_{m}p_{f}}|$,
\\ $min(|\overline{p_np_m}|, |\overline{p_mp_f}|) \leq \frac{D}{2}$ and  $max(|\overline{p_np_m}|, |\overline{p_mp_f}|) > \frac{D}{4}$
\item[\textit{A4P}] :$|\mathcal{P}(\mathcal{SS}_i)| =4$. And letting  $\mathcal{P}(\mathcal{SS}_i) =\{ p_n, p_{m1}, p_{m2}, p_f\}$,  \\
$|\overline{p_np_{m1}}| = |\overline{p_{m2}p_f}|\neq |\overline{p_{m1}p_{m2}}|$, 
$|\overline{p_np_{m1}}|, |\overline{p_{m2}p_f}| \leq \frac{D}{2}$ and  $|\overline{p_{m1}p_{m2}}| \geq \frac{D}{4}$
\end{enumerate}

\begin{figure*}
  \begin{center}
    \includegraphics[scale=0.8,width=8cm]{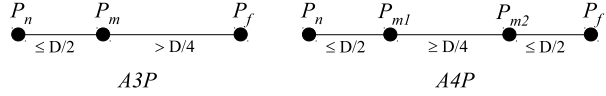}%
    \caption{Configurations in  \textit{A3P} and \textit{A4P}.}
    \label{fig:A3PA4P}
  \end{center}
\end{figure*}

\Newcodeline
\begin{algorithm}[h]
\caption{Reduce-Distance-LDS($r_i$)}
\label{algo:RDLDS}
{\footnotesize
\begin{tabbing}
111 \= 11 \= 11 \= 11 \= 11 \= 11 \= 11 \= \kill
{\em Assumptions}: non-rigid with $\delta$, $D \leq 2\delta$. \crm
{\em Input}:  $|\mathcal{P}(\mathcal{SS}_i)| =2$ and $\frac{D}{2} \leq |\overline{p_np_f}|$ \crm
\crm
\Cl \> {\bf if} $|\mathcal{P}(\mathcal{SS}_i)| =1$ {\bf then} no action \crm 
\Cl \> {\bf else if} $|\overline{p_fp_n}| \geq \frac{D}{2}$ {\bf then} \crm %
\Cl \>\> {\bf if} $|\mathcal{P}(\mathcal{SS}_i)| =2$ {\bf then} \crm %
\Cl \> \>\> {\bf if} $\frac{3D}{2}<|\overline{p_fp_n}| $ {\bf then} $des_i \leftarrow [p_n,p_f] + \frac{D}{2}$\crm
\Cl \>\> \> {\bf else if} $\frac{9D}{8} \leq  |\overline{p_fp_n}| \leq \frac{3D}{2}$ {\bf then} $des_i \leftarrow [p_n,p_f] + \frac{D}{12}$\crm
\Cl \> \> \>  {\bf else} $des_i \leftarrow  [p_n,p_f] + \frac{|\overline{p_fp_n}|}{2}-  \frac{3D}{16}$//$|\overline{p_fp_n}| < \frac{9D}{8}$ \crm
\Cl \>\>  {\bf else if} \textit{A3P} {\bf then}\crm
\Cl \>\>  \> {\bf if} $p_i=p_n$ {\bf and} $|\overline{p_ip_{m}}|  < |\overline{p_{m}p_{f}}|$ {\bf then } $des_i \leftarrow  p_m$ {\bf else} $des_i \leftarrow p_i$ \crm
\Cl \>\>  {\bf else if} \textit{A4P} {\bf then} \crm
\Cl \>\> \> {\bf if}  $p_i=p_n$ {\bf then} $des_i \leftarrow  p_{m_1}$ {\bf else} $des_i \leftarrow p_i$
\end{tabbing}
}
\end{algorithm}

In {\bf Algorithm~\ref{algo:RDLDS}}, expression $[p,q]+\alpha-\beta$ means the point from point $p$ with distance $\alpha-\beta$ on the line segment $\overline{pq}$.
If $D \leq 2\delta$, each moving in {\bf Algorithm~\ref{algo:RDLDS}} can reach the destination.
In {\bf Algorithm~\ref{algo:RDLDS}}, the final destination is a 2-point location such that the length of the $2$-point locations is at least  $D/8$ and less than $D/4$ from the initial 2-point location via configurations of  \textit{A3P}or \textit{A4P}. 
We will show that if any configuration obtained  in the algorithm are not $2$-point locations,  \textit{A3P}or \textit{A4P} holds for the configuration and therefore,
any $3$-point configuration appearing in {\bf Algorithm~\ref{algo:ILG}} does not appear. That is, {\bf Algorithm~\ref{algo:ILG}} cannot be performed during the execution of {\bf Algorithm~\ref{algo:RDLDS}}.
We have the following lemma. 

\begin{lemma}
\label{lemma:ILG}
Let $D \leq 2\delta$.
Without assuming chirality,
{\bf Algorithm~\ref{algo:RDLDS}} makes a configuration with only two positions on which all robots are located  and whose  distance $d$ satisfies $\frac{D}{4} \leq d < \frac{D}{2}$
from any configuration with only two positions on which all robots are located  and whose distance $d$ is at least $D/2$.
All configurations appearing in {\bf Algorithm~\ref{algo:RDLDS}} do not contain any 3-point location appearing in  {\bf Algorithm~\ref{algo:ILG}}.
\end{lemma}
\begin{proof} 
Let $d$ be the distance of the two point-location in the initial configuration. 
Note that since we assume $D \leq 2\delta$, each moving in the algorithm always reaches the destination.
In the case that $d>3D/2$, since robots located on the endpoints move inside by $D/2$ ({\bf line} 4), the configuration becomes two point-location whose distance reduces by $D/2$ or $D$, otherwise it satisfies \textit{A3P} or  \textit{A4P}. The configuration of \textit{A3P} or \textit{A4P} eventually becomes two point-location whose distance reduces by $D/2$ or $D$ ({\bf lines} 7-10).
This process is called {\em reducing process}.
Repeating these reducing processes, therefore, a configuration of two point-location whose distance $d'$ is greater than $D/2$ and at most  $3D/2$ is obtained.
Also any 3-point location appearing in  {\bf Algorithm~\ref{algo:ILG}} never appear in these reducing processes.

If $\frac{9D}{16} \leq d' \leq \frac{3D}{2}$,  repeating the reducing processes at most $5$ times, the distance becomes less than $\frac{9D}{8}$ and
it is at least $\frac{D}{2}$ at the first time it becomes less than $\frac{9D}{8}$.
Also \textit{A3P} or \textit{A4P} is satisfied every time this process is executed and any 3-point location appearing in  {\bf Algorithm~\ref{algo:ILG}} never appear.
If $d < \frac{9D}{8}$, repeating the redicing processes at most $4$ times, the distance becomes at least $\frac{D}{4}$ and less than $\frac{D}{2}$ and
\textit{A3P} or \textit{A4P} is satisfied every time {\bf line}~6 is executed and any 3-point location never appears.
Therefore, the lemma holds.
\end{proof}

\Newcodeline
\begin{algorithm}[t]
\caption{ElectLDS-Preserving-Distance($r_i$)}
\label{algo:ELDSP}
{\footnotesize
\begin{tabbing}
111 \= 11 \= 11 \= 11 \= 11 \= 11 \= 11 \= \kill
{\em Assumptions}: no-light, non-rigid with $\delta$, chirality \crm
{\em Input}:  $D$-distant configuration ($D \leq 2\delta$) \crm
{\em Subroutines}: 
\textit{Reduce\#LDS}, 
\textit{MakeDiameter},
\textit{MakeEdgeonBorder}.\crm
{\em Predicates}:
\textit{AFTER$-$RP(p)},
$\mathit{EDGEonBORDER}(\mathcal{SS}_i)$. \crm \crm
\Cl \> {\bf if} $|\mathit{LDS}(\mathcal{SS}_i)|>1$ {\bf then} \crm
\Cl \> \> {\bf if} $\mathcal{H}(\mathcal{SS}_i)$ is a regular $f$-polygon {\bf then} \crm
\Cl \> \> \>{\bf if not} $CLEAN(\mathcal{SS}_i)$  {\bf then} {\bf if}  $p_i \not \in \partial\mathcal{H}(\mathcal{SS}_i)$ {\bf then} $des_i \leftarrow CTR({\bf SS}_i)$ \crm
\Cl \> \> \> {\bf else if} $\textit{Diam}(\mathcal{SS}_i) \leq D$  {\bf then} $des_i \leftarrow CTR({\bf SS}_i)$\crm
\Cl \> \> \>\>\>{\bf else} {\bf if} $p_i \in \partial{\cal H}(\mathcal{SS}_i)$  {\bf then} $des_i \leftarrow [CTR({\bf SS}_i),p_i]- \frac{D}{2}$\crm
\Cl \> \> {\bf else} // $\mathcal{H}(\mathcal{SS}_i)$ is a non-regular $f$-polygon\crm
\Cl \> \> \>{\bf if} $\textit{Diam}(\mathcal{SS}_i) \leq D$ {\bf and} $AFTER$-$RP(p)$ {\bf then} $des_i \leftarrow p$\crm
\Cl \> \> \>{\bf else if}  $EDGEonBORDER(\mathcal{SS}_i)$ {\bf then}\crm
\Cl \> \> \>\>{\bf if} $\exists \overline{pq} \in \mathit{LDS}(\mathcal{SS}_i)$ s.t. $|\overline{pq}|= \textit{Diam}(\mathcal{SS}_i)$ {\bf then}\crm
\Cl \> \> \> \> \>{\bf if} $\forall p_j \exists q (\overline{p_jq} \in \mathit{LDS}(\mathcal{SS}_i))$  {\bf then} $Reduce\#LDS$ \crm
\Cl \> \> \> \>\> {\bf else}  {\bf if} $\forall q ( \overline{p_i q} \not \in \mathit{LDS}(\mathcal{SS}_i))$ {\bf then} $des_i \leftarrow CTR(\mathcal{SS}_i)$\crm
\Cl \> \> \> \>{\bf else if}  $p_i$ is a single-endpoint of \textit{LDS} {\bf then} $MakeDiameter$ \crm
\Cl \> \> \>{\bf else if} $\overline{p_iq} \in \mathit{LDS}(\mathcal{SS}_i)$ {\bf and} $q \in SEC(\mathcal{SS}_i)$ {\bf then} $MakeEdgeonBorder$ \crm
\Cl \> {\bf else} // $|\mathit{LDS}({\cal SS}_i)|=1$ \label{EPD-14}\crm
\Cl \> \> {\bf if} $p_i \neq p_n$ {\bf then} $des_i \leftarrow p_n$\label{EPD-15}
\end{tabbing}
}
\end{algorithm}

Lastly, the problem to be remained is to make the special configuration preserving the length of its \textit{LDS} is at least $D/2$
from any $D$-distant initial configuration.
Although \textit{ElectOneLDS} in Lemma~\ref{lemma:ElectLDS} makes the unique \textit{LDS} from any initial configuration,
it does not preserve the  length of its \textit{LDS} with   $d \geq D/2$ even if we assume $D$-distant initial configurations.
We will prove that {\bf Algorithm~\ref{algo:ELDSP}} performs this task from any $D$-distant configuration.

{\bf Algorithm~\ref{algo:ELDSP}} is based on \textit{ElectOneLDS} and produces the unique \textit{LDS} with its length at least $D/2$ unless it produces a Gathering configuration. 
If the unique \textit{LDS} is obtained, it produces a configuration with only two positions on which all robots are located and its distance $d$ is at least $D/2$ ({\bf lines}~14-15).

In {\bf Algorithm}~\ref{algo:ELDSP} the following  subroutines and predicates are utilized.

\medskip

\noindent{\bf Subroutines}

\vspace*{-4mm}

\begin{itemize}
\item \textit{Reduce\#LDS}: Let $q_0, \ldots, q_{g-1}$ be $g$ positions on $\mathit{SEC}(\mathcal{SS}_i)$ in counter-clockwise order.
Let $CS=\{ \overline{q_j q_{(j+1) \bmod g}} | |\overline{q_j q_{(j+1) \bmod g}}|= \min(|\overline{q_0q_1}| ,\ldots, |\overline{q_{g-1}q_0}|) \}$
If $ \overline{q_j q_{(j+1) \bmod g}} \in CS$, and $p_i = q_j$,
then $r_i$ moves from $p_i$ to $q=q_{(j+1) \bmod g}$ on arc 
$\overarc{p_i q}$\footnote{Since robots cannot move on arc, they move on the line segment of length $\delta$ repeatedly. That is, if $|\overline{p_i des_i}| \leq \delta$, then $r_i$ moves to the destination. Otherwise, the destination is the intersection of the circle with center $p_i$ and radius $\delta$, and $\mathit{SEC}(\mathcal{SS}_i)$.}. 

\item \textit{MakeDiameter}: If $p_i$ is the single-endpoint of \textit{LDS} and $q$ is another endpoint of the \textit{LDS},
let $des_i$ be the intersection of line passing $q$ and $CTR(\mathcal{SS}_i)$ and $\mathit{SEC}(\mathcal{SS}_i)$.
Note that the \textit{LDS} is not a diameter of $\mathit{SEC}(\mathcal{SS}_i)$.
Robot $r_i$ moves from $p_i$ to $des_i$ on arc $\overarc{p_i des_i}$  (Figure~\ref{fig:MakeDiameter}).

\item \textit{MakeEdgeonBorder}: If $\overline{p_iq} \in \mathit{LDS}(\mathcal{SS}_i)$ and $q \in \mathit{SEC}(\mathcal{SS}_i)$,
let $des_i$ be the intersection of half-line extending the line segment $\overline{p_iq}$ from $p_i$ and $\mathit{SEC}(\mathcal{SS}_i)$.
\end{itemize}

\begin{figure*}
  \begin{center}
    \includegraphics[scale=0.4,width=5cm]{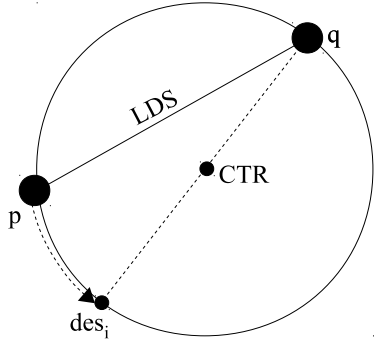}%
    \caption{Subroutine \textit{MakeDiameter}.}
    \label{fig:MakeDiameter}
  \end{center}
\end{figure*}

\noindent {\bf Predicates}

\vspace*{-4mm}

\begin{itemize}
\item $\mathit{CLEAN}(\mathcal{SS}_i)$: All robots are located on $CTR(\mathcal{SS}_i)$ or $\mathit{SEC}(\mathcal{SS}_i)$.
This configuration can be reached from Reg.Polygon ({\bf line} 3).
\item \textit{AFTER$-$RP(p)}: There exists a unique point $p$ such that all robots are located on $p$ and the circle with center $p$. 
This configuration can be reached after a regular polygon is broken and becomes $\textit{CLEAN}(\mathcal{SS}_i)$.
\item $\mathit{EDGEonBORDER}(\mathcal{SS}_i)$: Both endpoints of all \textit{LDS}s are located on $\mathit{SEC}(\mathcal{SS}_i)$.
\end{itemize}

\begin{figure*}
  \begin{center}
    \includegraphics[width=0.9\textwidth, keepaspectratio]{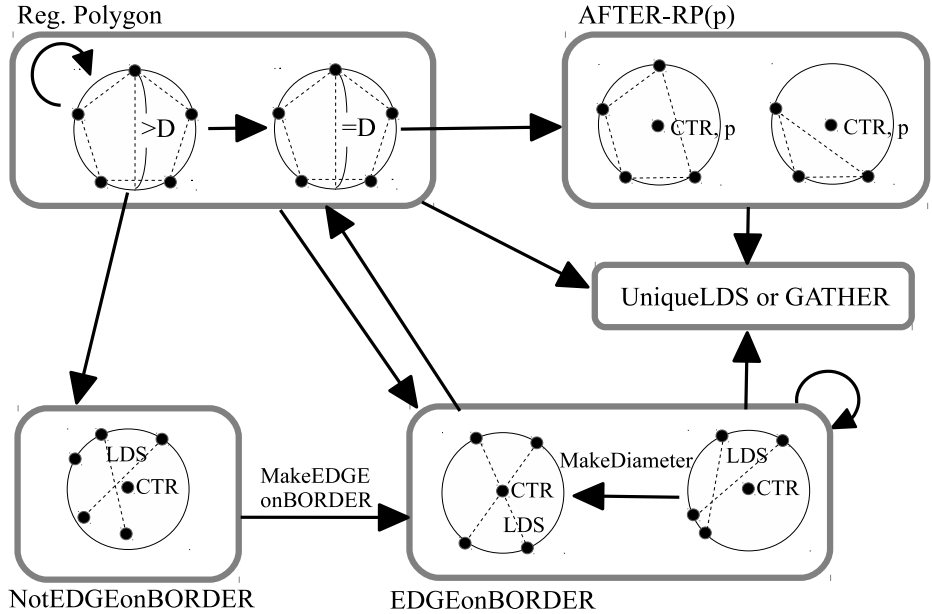} 
        \caption{The outline of ElectLDS-Preseving-Distance.}
    \label{fig:Int-light}
  \end{center}
\end{figure*}

The outline of {\bf Algorithm~\ref{algo:ELDSP}} is shown in Figure~\ref{fig:Int-light}. 
There are typical non-final configurations, Regular Polygon (\textit{Reg.Polygon}), Edges on the Border (\textit{EDGEonBORDER}) 
and Edges not on the Border (\textit{NotEDGEonBORDER}), where
\textit{Reg.Polygon} is a configuration the convex hull is a regular polygon,
\textit{EDGEonBORDER} is a configuration both endpoints of all \textit{LDS}s are located on the smallest enclosing circle and \textit{NotEDGEonBORDER} is a configuration both endpoints of some \textit{LDS} are not located on the smallest enclosing circle.

When $|\mathit{LDS}({\cal SS}_i)|=1$, {\bf Algorithm~\ref{algo:ELDSP}}  makes $2$-point configuration ({\bf line} 14).
Otherwise, if a configuration is Reg.Polygon whose diameter is more than $D$, 
{\bf Algorithm~\ref{algo:ELDSP}} makes a configuration of a regular polygon with diameter $D$, where robots are located on the corners and the center of the polygon (\textit{CLEAN}).  
In the process of making the regular polygon with diameter $D$ or after making the regular polygon, it makes a regular polygon with smaller number of corners, \textit{EDGEonBORDER}, or \textit{NotEDGEonBORDER}.
If the regular polygon with diameter $D$ is obtained, the algorithm
makes each robot located on a corner of the polygon move to the center of the \textit{SEC} and
makes a gathering configuration or a configuration with only one \textit{LDS} via configurations satisfying \textit{AFTER$-$RP(p)}, 
which is one that some corners are missing from the regular polygon. 

If \textit{EDGEonBORDER} is obtained and its \textit{LDS} is the diameter of the \textit{SEC} ({\bf lines}~8-9), 
the algorithm reduces the number of \textit{LDS}s after moving each robot which is located on a vertex of the polygon but is not an endpoint of any \textit{LDS} to the center of the polygon ({\bf  lines}~10-11).
If \textit{EDGEonBORDER} is obtained but its \textit{LDS} is not the diameter of the \textit{SEC}, its \textit{LDS} is changed into the diameter of the \textit{SEC} by moving the robot which is a single-endpoint of \textit{LDS} ({\bf line}~12). 
Existing such robots is guaranteed by Lemma~\ref{lemma:LDSofCH7}.

If \textit{NotEDGEonBORDER}  is obtained, the robot not located on the \textit{SEC} is moved on the \textit{SEC} ({\bf line}~13).

Let $\mathcal{C}(t)$ be a configuration at time $t$ in {\bf Algorithm}~\ref{algo:ELDSP}.
If  $\mathcal{H}(\mathcal{C}(t))$ is a regular $f$-polygon and $\mathit{CLEAN}(\mathcal{C}(t))$ does not hold,
the number of robots located on \textit{CTR}($\mathcal{C}(t+1)$) increased ({\bf line}~3) and 
therefore, there exists a time $t'$ such that $\mathcal{H}(\mathcal{C}(t'))$ is a regular $f$-polygon and $\mathit{CLEAN}(\mathcal{C}(t'))$ holds.

\begin{lemma}\label{lemma_case_regpolyGTD}
If $\mathcal{H}(\mathcal{C}(t))$ is a regular $f$-polygon, $\mathit{Diam}(\mathcal{C}(t)) > D$ and $\mathit{CLEAN}(\mathcal{C}(t))$ holds, one of the followings is satisfied.
\begin{enumerate}
\item $\mathcal{H}(\mathcal{C}(t))=\mathcal{H}(\mathcal{C}(t+1))$ but the number of robots located on \textit{SEC}$(\mathcal{C}(t+1))$ decreases.
\item $\mathcal{H}(\mathcal{C}(t+1))$ is the regular $f$-polygon with $\mathit{Diam}(\mathcal{C}(t+1)) = D$.
\item $\mathcal{H}(\mathcal{C}(t+1))$ is a regular $f'$-polygon such that  $f' \leq f$ and $\mathit{Diam}(\mathcal{C}(t)) \geq \mathit{Diam}(\mathcal{C}(t+1)) > D$.
\item $\mathcal{H}(\mathcal{C}(t+1))$ is a non-regular $f'$-polygon such that $f' \leq f$ or $\mathit{Diam}(\mathcal{C}(t))>\mathit{Diam}(\mathcal{C}(t+1)) \geq D$.
\end{enumerate}
\end{lemma}
\begin{proof}
In this case, the destination of robot $r_i$ located on $\partial\mathcal{H}(\mathcal{C}(t))$ is the point of distance $\frac{D}{2}$ from $\mathit{CTR}(\mathcal{C}(t))$ on the line segment $\overline{p_i \mathit{CTR}(\mathcal{C}(t))}$ (Figure~\ref{fig_Lemma23_23}). 
If all vertices on $\partial\mathcal{H}(\mathcal{C}(t))$ remains, 1. holds.
If all robots move to the destinations, 2. holds. Otherwise, 3. or 4. holds according to movements of robots.
In the both case, the diameter of the polygon decreases but is at least $D$.  
\end{proof}

\begin{figure*}
\begin{center}
    \includegraphics[width=0.4\textwidth, keepaspectratio]{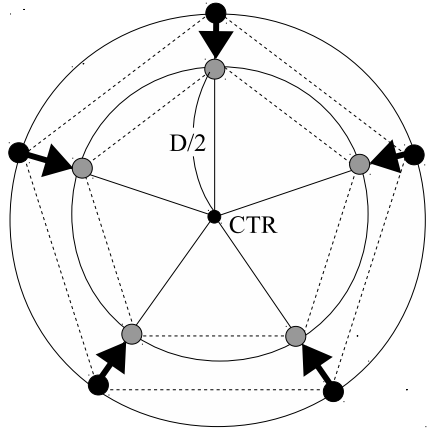}
\end{center}
\caption{Movements of robots in Lemma~\ref{lemma_case_regpolyGTD}.}\label{fig_Lemma23_23}
\end{figure*}

By Lemma~\ref{lemma_case_regpolyGTD}, there exists a time $t$ such that $\mathcal{H}(\mathcal{C}(t))$ is a regular $f$-polygon ($f \geq 2$), $\mathit{Diam}(\mathcal{C}(t)) = D$ and \textit{CLEAN}($\mathcal{C}(t))$ holds. If $f=2$, then $\mathit{LDS}(\mathcal{C}(t))=1$ and 2-point configuration is attained  
with $\mathit{Diam}(\mathcal{C}(t)) =D$ ({\bf line}~14-15).

\begin{lemma}
If $\mathcal{H}(\mathcal{C}(t))$ is a regular $f$-polygon ($f \geq 3$), $\mathit{Diam}(\mathcal{C}(t)) = D$ and $\textit{CLEAN}(\mathcal{C}(t))$ holds, one of the followings is satisfied, and when $f \neq 3m (m>1)$, 4. and 5. never occur. 
\begin{enumerate}
\item $\mathcal{H}(\mathcal{C}(t))=\mathcal{H}(\mathcal{C}(t+1))$ but the number of robots located on $\mathit{SEC}(\mathcal{C}(t+1))$ decreases.
\item $\mathcal{C}(t+1)$ attains Gathering, or $\mathit{LDS}(C(t+1))=1$ such that $\mathit{Diam}(\mathcal{C}(t+1)) \geq D/2$.
\item $\mathcal{C}(t+1)$ satisfies  \textit{AFTER$-$RP(p)} and $p=\mathit{CTR}(\mathcal{C}(t))$.
\item $\mathcal{H}(\mathcal{C}(t+1))$ is an isosceles triangle with base of $D/2$.
\item $\mathcal{H}(\mathcal{C}(t+1))$ is a rhombus with two isosceles triangles with base of $D/2$.
\end{enumerate}
\end{lemma}
\begin{proof}
In this case, the destination of robot $r_i$ located on $\partial\mathcal{H}(\mathcal{C}(t))$ is \textit{CTR}$(\mathcal{C}(t))$.
Since $D \leq 2 \delta$, active robots always reach the destination. 
If all points on $\partial\mathcal{H}(\mathcal{C}(t))$ remains, 1. holds. If all robots, robots located on all points except one point, or all points except two points constituting a diameter of  $\mathit{SEC}(\mathcal{C}(t))$ move to the destination, 2. holds.

Otherwise, 
as long as $f \neq 3m (m>1)$ or $3$, any $2$ points on $\partial\mathcal{H}(\mathcal{C}(t))$ and $\mathit{CTR}(\mathcal{C}(t))$ do not constitute 
an isosceles triangle with base of $D/2$. Thus, it is easily verified that 3. holds and 4. and 5. are not satisfied.
In the case that $f = 3m (m>1)$,   if 4. and 5. are not satisfied, it is easily verified that 3. holds.
\end{proof}

The following lemmas are easily verified.

\begin{lemma}
Assume that $\mathcal{H}(\mathcal{C}(t))$ is a non-regular polygon, $\mathit{Diam}(\mathcal{C}(t)) = D$ and \textit{AFTER$-$RP(p)} holds. 
Then one of the followings holds
\begin{enumerate}
    \item  $\mathit{Diam}(\mathcal{C}(t+1)) = D$ and \textit{AFTER$-$RP(p)} still holds, 
    where $p=\mathit{CTR}(\mathcal{C}(t))$, but the number of robots located on $\mathit{SEC}
    (\mathcal{C}(t+1))$ decreases.
    \item  $\mathcal{C}(t+1)$ is a Gathering configuration.
    \item  $|\mathit{LDS}(\mathcal{C}(t+1))|=1$ and 
    $\mathit{Diam}(\mathcal{C}(t+1)) \geq \frac{D}{2}$.
    \item   $\mathcal{H}(\mathcal{C}(t+1))$ is an isosceles triangle with 
    base of $D/2$ or a rhombus with two isosceles triangles with base of $D/2$.
\end{enumerate}
\end{lemma}

\begin{lemma}
If $\mathcal{H}(\mathcal{C}(t))$ is an isosceles triangle with base of $D/2$ or a rhombus with two isosceles triangles with base of $D/2$,
there is a time $t' \geq t$ such that $\mathcal{C}(t')$ attains Gathering or $|\mathit{LDS}(\mathcal{C}(t'))|=1$ and $\mathit{Diam}(\mathcal{C}(t')) \geq D/2$.
\end{lemma}

The remaining case is that $\mathcal{H}(\mathcal{C}(t))$ is a non-regular polygon and $\mathit{Diam}(\mathcal{C}(t))>D$.

\begin{lemma}
If $|\mathit{LDS}(\mathcal{C}(t))| \geq 2$, $\mathcal{H}(\mathcal{C}(t))$ is a non-regular $f$-polygon, $\overline{pq} \in \mathit{LDS}(\mathcal{C}(t))$ and $|\overline{pq}| =\mathit{Diam}(\mathcal{C}(t))$, then
there exists a time $t' (>t)$ such that $\mathcal{H}(\mathcal{C}(t'))$ is a non-regular $f'$-polygon and $\textit{Diam}(\mathcal{C}(t'))=\textit{Diam}(\mathcal{C}(t))$ but $\mathit{LDS}(\mathcal{C}(t')) < \mathit{LDS}(\mathcal{C}(t))$.
\end{lemma}
\begin{proof}
In this case, {\bf Algorithm~\ref{algo:ELDSP}} performs \textit{Reduce\#LDS} after any robot located on a point which is not an endpoint of \textit{LDS}
moves to $\mathit{CTR}(\mathcal{C}(t))$({\bf line} 11). Let $CS$ be a set of the minimum segments defined in \textit{Reduce\#LDS} and let $s_j=\overline{q_jq_{(j+1) \bmod g}}\in CS$. Robots located on endpoints of $s_j$ may be going to move in parallel if they become active in this round. Since the destinations of these robots
is uniquely determined due to the chirality, the length of $s_j$ decreases by $\delta$. Repeating the process, since $s_j$ is a segment that both points are endpoints of a diameter of 
$\mathit{SEC}(\mathcal{C}(t))$, $|\mathit{LDS}(\mathcal{C}(t))|$ decreases.
Thus, the lemma holds.
\end{proof}

\begin{lemma}
If $\mathcal{H}(\mathcal{C}(t))$ is a non-regular polygon, \textit{EDGEonBORDER}($\mathcal{C}(t))$ holds and $|\overline{pq}| <\mathit{Diam}(\mathcal{C}(t))$ for any $\overline{pq} \in \mathit{LDS}(\mathcal{C}(t))$,
there exists a time $t'>t$ such that
there exist $p'$ and $q'$ such that $|\overline{p'q'}| =\mathit{Diam}(\mathcal{C}(t'))$, \textit{EDGEonBORDER}($\mathcal{C}(t'))$ holds and $\mathit{Diam}(\mathcal{C}(t')) = \mathit{Diam}(\mathcal{C}(t))$.
\end{lemma}
\begin{proof}
In this case, {\bf Algorithm~\ref{algo:ELDSP}} performs \textit{MakeDiamter}. Since there exists a single-endpoint of \textit{LDS} whose endpoints
are located on $\mathit{SEC}(\mathcal{C}(t))$ and is not a diameter of  $\mathit{SEC}(\mathcal{C}(t))$ by Lemma~\ref{lemma:LDSofCH7}, let $p_0$ be a single-endpoint of such \textit{LDS} of $\overline{p_0q_0}$. The destination of the robot $r_i$ is located on $p_0$ is the intersection of line passing $q_0$ and $CTR(\mathcal{C}(t))$ and $\mathit{SEC}(\mathcal{C}(t))$. Since there do not exist points on the arc $\overarc{p_0 des_i}$ except $p_0$ by Lemma~\ref{lemma:LDSofCH6}, the robot $r_i$ can move on the arc $\overarc{p_0 des_i}$ and can reach the destination. Then $\overline{des_iq_0}$ becomes the diameter of the \textit{SEC}. If there is more than one single-endpoint of \textit{LDS}, these robots on the single-endpoints can move independently by using Lemma~\ref{lemma:LDSofCH6}.
Therefore, the lemma holds.
\end{proof}

\begin{lemma}\label{lemma-notEoB}
If \textit{EDGEonBORDER}$(\mathcal{C}(t))$ does not hold, then there exists a time $t' (>t)$ such that \textit{EDGEonBORDER}$(\mathcal{C}(t'))$ holds and
 $\textit{Diam}(\mathcal{C}(t')) = \textit{Diam}(\mathcal{C}(t))$.
\end{lemma}
\begin{proof}
In this case, {\bf Algorithm}~\ref{algo:ELDSP} performs \textit{MakeEdgeonBorder}. It can verified that there exists a time $t'$ such that both endpoints of all \textit{LDS}s 
are located on $\textit{SEC}(\mathcal{C}(t'))$ and $\textit{Diam}(\mathcal{C}(t')) = \textit{Diam}(\mathcal{C}(t))$
\end{proof}

We obtain the following theorem by Lemma~\ref{lemma:ILG}-\ref{lemma-notEoB}.

\begin{theorem}
Let $D \leq 2\delta$.
Gathering is solvable in internal-light, non-rigid with $\delta$, and SSYNC, if robots have $2$ lights, set-view and agreement of chirality, and the initial configuration is $D$-distant.
\end{theorem}

\section{Gathering Algorithms in CENT}
\label{sec:CSET}

In this section, we construct Gathering algorithms in CENT scheduler.
Theorem~\ref{theorem:centrr} implies that Gathering cannot be solved 
even in CENT and $2$-BOUNDED, and distinct Gathering cannot be solved 
even in ROUND-ROBIN without strong multiplicity. 
Although it is still open whether distinct Gathering can be solved in 
ROUND-ROBIN without lights, we show that lights with two colors are
enough to solve Gathering in CENT or ROUND-ROBIN with weaker assumptions than
those of SSYNC as follows.

\begin{enumerate}
\item[(e)] Algorithm7:(2, external, CENT, non-rigid), and
\item[(f)] Algorithm8:(2, internal, ROUND-ROBIN, rigid).
\end{enumerate}

\subsection{CENT scehdulers}

\Newcodeline
\begin{algorithm}[h]
\caption{Cent-Ext-Light-Gather($r_i$)}
\label{algo:CELG}
{\footnotesize
\begin{tabbing}
111 \= 11 \= 11 \= 11 \= 11 \= 11 \= 11 \= \kill
{\em Assumption}: External-light, 2 colors($T$ and $M$), non-rigid, set-view, CENT \crm
{\em Input}: Any initial configuration and any color of each robot's light. \crm
\Cl \> {\bf if} $T \not \in \mathcal{L}(\mathcal{SS}_i)$ {\bf then} $\ell_i \leftarrow T$\crm 
\Cl \>  {\bf else if} $|\mathcal{P}_{T}(\mathcal{SS}_i)| = 1$ {\bf then} 
$\ell_i \leftarrow M$; $des_i \leftarrow p(\in  {\cal P}_{T}({\cal SS}_i))$\crm
\Cl \> {\bf else}  $\ell_i \leftarrow M$
\end{tabbing}
}
\end{algorithm}

{\bf Algorithm~\ref{algo:CELG}} is a Gathering algorithm in external-light, non-rigid and CENT. 
This algorithm is self-stabilizing and does not use the local-aware assumption.
It uses two colors $T$\textit{(Target)} and $M$\textit{(Move)}. It does not use \textit{ElectOneLDS} and makes a configuration where there is just one position having some robots with $T$.

Since {\bf Algorithm~\ref{algo:CELG}} is executed in CENT, if each active robot observes at least one position having robots with $T$,
the robot changes its color to $M$ ({\bf lines} 2-3), and the configuration can reach to a configuration of one position (gathering position) having robots with $T$.
Once this configuration is obtained, the configuration is unchanged and robots with $M$ not located on the gathering position move to the position, and robots located on the gathering position stay there and change their color to $T$, since these robots do not observe any position having robots with $T$ due to the local-unawareness.  If an initial configuration has no robots with $T$, the algorithm makes the light of a robot being active at this round color $T$ ( {\bf line} 1). Since we assume external-light, this robot does not know whether its own color is $T$ or $M$. However, the color of the robot becomes $T$ in either case. It can be easily verified this algorithm is performed in non-rigid movement and self-stabilizing. Thus, we have the following theorem. 

\begin{theorem}
Without chirality and local-awareness,
Cent-Ext-Light-Gather is a self-stabilizing gathering algorithm in external-light, non-rigid and CENT, with $2$ colors and set-view.
\end{theorem}

If we assume rigid movement and distinct configurations,
{\bf Algorithm~\ref{algo:CELG}} can be transformed into an algorithm performing with arbitrary-view.  When the number of positions having robots with $T$ becomes one and
robots observe the position, they change their color $T$ and move to the position.  It can be easily shown that this modified algorithm solves distinct Gathering with arbitrary-view since robots with different colors do not occupy the same location during the execution. 

\begin{corollary}
Without chirality and local-awareness,
distinct Gathering is solvable in external-light, rigid, and CENT with $2$ colors and arbitrary-view.
\end{corollary}

\subsubsection{ROUND-ROBIN Scheduler}

If ROUND-ROBIN scheduler is assumed, Gathering can be done by {\bf Algorithm~\ref{algoRRILG}} in internal-light and rigid with $2$ colors.
In this case, we also use \textit{ElectOneLDS} to obtain configuration \textit{onLDS}.
After \textit{onLDS},
{\bf Algorithm~\ref{algoRRILG}} makes a configuration of $2$ locations $p_n$ and $p_f$ (denoted by $AA$) or $3$ locations $p_n, p_m, p_f$ such that  $(|\overline{p_f p_m}|=|\overline{p_m p_n}|))$ (denoted by $AAA$), in either case  all robots with $A$ are located on these locations.
For these cases that the configurations $AA$ or $AAA$, the following lemmas hold by using ROUND-ROBIN property.

\Newcodeline
\begin{algorithm}[h]
\caption{RR-Int-Light-Gather($r_i$)}
\label{algoRRILG}
{\footnotesize
\begin{tabbing}
111 \= 11 \= 11 \= 11 \= 11 \= 11 \= 11 \= \kill
{\em Assumption}: Internal-light, 2 colors($A$ and $B$), rigid, set-view, ROUND-ROBIN \crm
{\em Input}:  \textit{onLDS} , all robots have color $A$.\crm
\Cl \> {\bf case} $\ell_i$  {\bf of } \crm
\Cl \> A:\crm %
\Cl \> \> {\bf if} $(|{\cal P}({\cal SS}_i)| =2)$ {\bf or}  $((|{\cal P}({\cal SS}_i)| =3)$ {\bf and} $(|\overline{p_f p_m}|=|\overline{p_m p_n}|))$ {\bf then}\crm
\Cl \> \> \> $\ell_i \leftarrow B$; $des_i \leftarrow \frac{p_n+p_f}{2} $ \crm
\Cl \> \> {\bf else} $des_i \leftarrow p_n$ \crm %
\Cl \> B:\crm %
\Cl \> \>{\bf if} $|{\cal P}({\cal SS}_i)| =2$ {\bf then} $des_i \leftarrow p_f$\crm
\Cl \> \>{\bf else} $des_i \leftarrow \frac{p_n+p_f}{2}$// $|{\cal P}({\cal SS}_i)| =3$\crm
\Cl \> {\bf endcase} 
\end{tabbing}
}
\end{algorithm}

\begin{lemma}\label{caseAA}
Let $AA$ be a 2-point configuration whose locations are $p$ and $q$ and all robots have color $A$.
For the activation order of robots, let $last_{p}$ and $last_{q}$ be the maximum of orders among robots located on $p$ and $q$, respectively, and let $min=min(last_p,last_q)$. 
After $n$ rounds of the executions of {\bf Algorithm}~\ref{algoRRILG},
the configuration becomes $BB$, all robots on the location having the robot with order of $min$ have the orders less than or equal to $min$ and all robots on the other location have the orders more than $min$. 
\end{lemma}
\begin{proof}
Let $p'$ and $q'$ be  the locations of $BB$.
Consider the case that $last_p < last_q$. In this case  $min=last_p$.
Since robots with color $A$ move to the midpoint of $p$ and $q$ (denoted as $r$) and change their color into $B$ ({\bf lines} 3-4),
the configuration becomes $BA$ whose locations are $r$ and $q$ after $last_p$ rounds and the orders of all robots located on $r$ ($q$) is less than or equal to $last_p$ (is more than $last_p$). 
Then after $n-last_p$ rounds robots located on $q$ are activated and the configuration becomes $BB$ whose locations are $r$ and the midpoint of $r$ and $q$. Thus the lemma holds.
The case that  $last_p < last_q$ can be proved similarly.
\end{proof}

\begin{lemma}\label{caseAAA}
Let $AAA$ be a 3-point configuration whose locations are $p$, $r$ and $q$ from the leftmost and all robots have color $A$.
For the activation order of robots, let $last_p$, $last_r$ and $last_q$ be the maximum of orders among robots located on $p$, $r$ and $q$, respectively, and let $min$ and $min_2$ be the minimum  and the second minimum of $last_p, last_r$, and  $last_q$, respectively. 
After $n$ rounds of the executions of {\bf Algorithm~\ref{algoRRILG}},
the configuration becomes $BB$ or $BBB$ and the followings hold.
\begin{enumerate}
\item In the case of $BB$, all robots on the location having the robot with order of $min$ have the orders less than or equal to $min$ and all robots on the other location have the orders more than $min$.
\item In the case of $BBB$, the location having the robot with order of $min$ is either left side or right side,
and all robots on the location having the robot with order of $min$ have the orders less than or equal to $min$, 
all robots on the other side have the orders between $min+1$ and $min_2$, and
all robots on the midpoint  have the orders more than $min_2$.
\end{enumerate}
\end{lemma}
\begin{proof}
Consider that case that $min=last_p$.
After $last_p$ rounds from the beginning, the configuration becomes $(A+B)A$
\footnote{$A+B$ means that robots with color $A$ and robots with color $B$ are located on the same location.}
whose locations are $r$ and $q$ and the orders of all robots with color $B$ located on $r$ are less than or equal to $last_p$, and the orders of all robots with color $A$ located on $r$ are between $last_p+1$ and $min_2$ ({\bf lines}~3-4).

If $min_2=last_r$, since the next activated robots are ones with color $A$, after the next $last_r-last_p$ rounds, the configuration becomes $BBA$ where, these locations are $r$, the midpoint of $r$ and $q$ (denoted as $r'$) and $q$, all robots located on $r$ have the orders less than or equal to $last_p$, all robots located on $r'$ have the order between $last_p+1$ and $last_r$ ({\bf lines}~3-4).
Then the next activated robots are ones located on $q$, and after the next $n-last_r$ rounds, 
the configuration becomes $BB$ ({\bf line}~8) whose locations are $r$ and $r'$ and the lemma holds.
Otherwise ($min_2=last_q$), after the next $last_q-last_p$ rounds, the configuration becomes $(A+B)B$ whose locations $r$ and $r'$ and all robots with color $B$ located on $r$ have the orders less than or equal to $last_p$, all robots located on $r'$ have the order between $last_p+1$ and $last_q$ ({\bf lines}~3-4). 
Then since the next activated robots are ones with color $A$, after the next $n-last_q$ rounds 
the configuration becomes $BBB$ whose locations $r$, the midpoint of $r$ and $r'$ (denoted as $r''$) and $r'$ ({\bf lines}~3-4) and  the lemma holds.

The case that $min=last_q$ can be proved similarly.

Consider the last case that $min=last_r$. After $last_r$ rounds from the beginning,
the configuration becomes $ABA$ whose locations are $p$, $r$ and $q$ and all robots with color $B$ have the orders
less than or equal to $last_r$ ({\bf lines}~3-4). 
If $min_2=last_p$, after the next $last_p-last_r$ rounds, the configuration becomes $BA$ whose locations are $r$ and $q$ and all robots on the location $r$ have the orders less than or equal to $last_p$ ({\bf lines}~3-4).
Then after the next $n-last_p$ rounds, the configuration becomes $BB$ whose locations $r$ and $r'$ ({\bf lines}~3-4) and the lemma holds.
The case that  $min_2=last_q$ is the symmetrical one and so can be proved similarly. 
\end{proof}

We have the following theorem by using Lemmas~\ref{caseAA} and \ref{caseAAA}.

\begin{theorem}
Without chirality, Gathering is solvable in internal-light, rigid, and ROUND-ROBIN, if robots have $2$ lights and set-view.
\end{theorem}
\begin{proof}
By using Lemmas~\ref{caseAA} and \ref{caseAAA}, the configuration becomes $BB$ or $BBB$ satisfying the conditions of Lemmas~\ref{caseAA} and \ref{caseAAA}.

In the case of $BB$, since all robots on the location having the robot with order of $min$ have the orders at most $min$ and all robots on the other location have the orders more than $min$, the configuration becomes a Gathering one after $min$ rounds from the configuration $BB$ ({\bf line}~7).

In the case of $BBB$, let $p$, $q$ and $r$ be the location having the robot with order of $min$, the other endpoint and the midpoint, respectively. since all robots located at $p$, $q$ and $r$ have the orders at most $min$, between $min$ and $min_{2}$, and more than $min_{2}$, respectively,  the configuration becomes $BB$ after $min$ rounds ({\bf line}~8), from which the Gathering attains after more $min_2$ rounds ({\bf line}~7).
\end{proof}

\section{Concluding Remarks}
\label{sec:conclusion}

We have shown Gathering algorithms by mobile robots with lights in SSYNC and CENT schedulers and
we have obtained some relationship between the power of lights (full, external and internal) and assumptions of robots' moving (rigid and non-rigid).
Interesting open questions are determining the relationship of precise computational power between external-light and internal-light and constructing Gathering algorithms in ASYNC.
For full-light, since it is known that any algorithm of robots in full-light of $k$ colors and SSYNC can be simulated by robots in full-light of $5k$ colors and ASYNC \cite{DFPSY}, Gathering is solvable in full-light, non-rigid, and ASYNC, if robots have $10$ colors, set-view and agreement of chirality by using the result in this paper. Reducing the number of colors in full-light and Gathering algorithms with external-light or internal-light in ASYNC remains as interesting open problems.  


\end{document}